\documentclass[a4paper,11pt]{article}

\usepackage[colorlinks=true]{hyperref}

\hypersetup{
    colorlinks,
    citecolor=green,
    filecolor=black,
    linkcolor=red,
    urlcolor=blue
    outlinecolor=black
}

\usepackage{amsmath}
\usepackage{amssymb}
\usepackage{ stmaryrd }
\usepackage{minitoc}
\usepackage{a4wide}
\usepackage{amsfonts}
\usepackage{amsmath,amsfonts,amssymb,amsthm}
\usepackage{mathrsfs}
\usepackage{url}
\usepackage{graphicx}
\usepackage{subfigure}
\usepackage[mathscr]{euscript}
\usepackage[usenames,dvipsnames]{color}

\newtheorem{thm}{Theorem}

\newtheorem{pr}{Proposition}

\newtheorem{lem}{Lemma}

\newtheorem{rk}{Remark}
\newcommand{\tr}[1]{\text{Tr}\left(#1\right)}
\newcommand{\trr}[1]{\text{Tr}^{2}\left(#1\right)}
\newcommand{\EE}[1]{{\text{\normalsize$\mathbb E$}}\left(#1\right)}
\newcommand{\half}{\text{\scriptsize $\frac{1}{2}$}}
\newcommand{\CC}{{\mathbb C}}
\newcommand{\bra}[1]{\left<#1\right|}
\newcommand{\ket}[1]{\left|#1\right>}
\newcommand{\bket}[1]{\left<#1\right>}
\newcommand{\GG}{\left| G\right>}
\newcommand{\Pgxi}{P_{\!\!g_\xi}}

\newcommand{\Pexi}{P_{\!\!e_\xi}}

\newcommand{\Spxi}{\text{\Large$\sigma$}_{\!\!+_\xi}}

\newcommand{\Smxi}{\text{\Large$\sigma$}_{\!\text{-}_\xi}}

\newcommand{\pb}{\bar{p}}
\newcommand{\pbp}{\bar{p}^{+}}
\newcommand{\pbm}{\bar{p}^{{-}}}
\newcommand{\Gxi}{\left| G_{\xi}\right>}

\newcommand{\Lc}{\mathcal{L}}
\newcommand{\Mc}{\boldsymbol{\mathcal{M}}^{n}}
\newcommand{\Ac}{\mathcal{A}}
\newcommand{\Htot}{H_{\text{\tiny tot}}}
\newcommand{\Id}{\mathbb{I}}
\newcommand{\etab}{\overline{\eta}}
\newcommand{\thetab}{\overline{\theta}}

\newcommand{\etaI}{\eta^{\infty}}
\newcommand{\etaT}{\widetilde{\eta}}
\newcommand{\MK}[1]{\tfrac{#1}{\tr{#1}}}

\title{Stability of  continuous-time quantum filters with measurement imperfections}

\author{Hadis~Amini
\thanks{Edward L. Ginzton Laboratory, Stanford University, Stanford, CA 94305, USA, {\tt nhamini@stanford.edu}.}
\and  Cl\'{e}ment~Pellegrini
\thanks{Institut de Math\'{e}matiques, IMT, Universit\'{e} de Toulouse (UMR 5219), 31062 Toulouse, Cedex 9, France, {\tt clement.pellegrini@math.univ-toulouse.fr}.}
\and Pierre~Rouchon
\thanks{Centre Automatique et Syst\`{e}mes, Mines ParisTech, 60 Bd Saint Michel, 75272 Paris cedex 06, France, {\tt\small pierre.rouchon@mines-paristech.fr}.}}

\date{December 2, 2013}
\begin{document}
\maketitle
\begin{abstract}
The fidelity between the state of a  continuously  observed quantum system and  the state of  its   associated quantum filter,  is shown to be always  a submartingale. The observed system is  assumed to be governed by a   continuous-time Stochastic Master Equation (SME),  driven simultaneously by Wiener and Poisson processes and  that takes into account incompleteness and errors in measurements. This stability result is the continuous-time counterpart of  a similar stability result already established for discrete-time quantum systems and where the measurement imperfections are modeled by a left stochastic matrix.
\end{abstract}
\section{Introduction}

Since the work of Davies~\cite{davies1976quantum}, the time evolution of the   state  (density operator) $\rho_t$ at time $t$ of an observed quantum system  can be described  by  a Stochastic Master Equation (SME)  taking into account the back-action of the measurements on $\rho_t$. Such  SMEs (see \cite{Barchiellei1} for a modern  exposure)  have been the starting point of the seminal contributions  of Belavkin to   quantum filtering and control~\cite{Belavkinone,belavkin1992quantum,Barchielli3,Belavkin2007}.  Quantum filters are used to get  an  estimate  $\rho^e_t$ of  $\rho_t$ based on an initial guess  $\rho^e_0$ of $\rho_0$ and on  the measurement outcomes  between $0$ and $t$. Quantum filtering  is    related to quantum trajectories~\cite{dalibard-et-al:PRL92,carmichael-book} and their original  motivations   for Monte-Carlo simulations. Roughly speaking,  quantum filtering replaces the  random numbers  used at each time-step of a  Monte-Carlo simulation by the measurement outcomes to update  the estimate  $\rho_t^e$.   An important practical issue deals with  the convergence:  does  $\rho^e_t$ converge towards $\rho_t$  when  $t$ tends to $+\infty$,  even if  $\rho^e_0\neq \rho_0$ ?  Few convergence results are available up to now, except the sufficient conditions established in~\cite{Ramonthesis,Ramon2009}  for diffusive SMEs. Particular results for quantum non demolition indirect measurements have been established in~\cite{Bauer2012,Benoist2013} As far as we know, general and checkable  necessary and sufficient convergence conditions do not exist yet.

In this paper we do not investigate directly convergence issues. We  focus on  stability issues to ensure the absence of time divergence. We prove in Theorem~\ref{thm:pws} that for a large class of  continuous-time SMEs driven simultaneously by Wiener and Poisson processes,  associated quantum filters are stable: the fidelity between $\rho$ and its estimate $\rho^e$,
\begin{equation}~\label{eq:fidelity}
F(\rho,\rho^e)=\trr{\sqrt{\sqrt{\rho}\rho^e\sqrt{\rho}}},
\end{equation}
is a sub-martingale.
In the Wiener case, Theorem~\ref{thm:pws}   shows that, whatever the Hamiltonian $H$,  the measurement operators $V_\nu$  and the detection efficiencies $\etab_\nu\in|0,1]$ are,
 the fidelity $F(\rho_t,\rho^e_t)$  is a sub-martingale where
 $\rho_t$ obeys the following  diffusive SME
\begin{multline*}
d\rho_t =-i[H,\rho_t] \,dt+ \left(\sum_\nu V_\nu \rho_t V_\nu^\dag - \half(V_\nu^\dag V_\nu\rho_t +\rho_t V_\nu^\dag V_\nu ) \right)\,dt
\\+\sum_{\nu} \sqrt{\etab_\nu}\left( V_\nu\rho_t + \rho_t V_\nu^\dag - \tr{(V_\nu+V_\nu^\dag)\rho_t} \rho_t \right) \,dW_\nu(t)
\end{multline*}  driven by the Wiener processes $W_\nu$, and where  the estimate $\rho^e_t$  obeys the following non-linear stochastic equation
\begin{multline*}
d\rho^e_t =-i[H,\rho^e_t] \,dt+\left(\sum_\nu V_\nu\rho_t^e V_\nu^\dag - \half(V_\nu^\dag V_\nu\rho_t^e +\rho_t^e V_\nu^\dag V_\nu ) \right)\,dt
\\
+\sum_{\nu} \sqrt{\etab_\nu}
 \left( V_\nu\rho_t^e + \rho_t^e V_\nu^\dag - \tr{(V_\nu+V_\nu^\dag)\rho_t^e} \rho_t^e\right)
\left(dy_\nu(t)-   \sqrt{\etab_\nu}\tr{(V_\nu+V_\nu^\dag)\rho^e_t}\,dt\right).
\end{multline*}
driven by the  measures  $dy_\nu(t)=dW_\nu(t)+\sqrt{\etab_\nu} \tr{(V_\nu+V_\nu^\dag)\rho_t}\,dt$ and initialized to any  density matrix $\rho^e_0$.
In the Poisson case, Theorem~\ref{thm:pws}  ensures, as for the Wiener case,   the stability  of the  quantum filtering process. Additionally, it  provides a new kind of SMEs taking into account incompleteness  and errors in the  jump detections:
\begin{multline*}
d\rho_t =-i[H,\rho_t] \,dt+ \left(\sum_\mu V_\mu\rho_t V_\mu^\dag - \half(V_\mu^\dag V_\mu\rho_t +\rho_t V_\mu^\dag V_\mu) \right)\,dt
\\
+\sum_{\mu}
   \left(\frac{\thetab_\mu\rho_t+\sum_{\nu}\etab_{\mu,\nu}V_{\nu}\rho_t V_\nu^\dag}{\thetab_\mu+\sum_{\nu}\etab_{\mu,\nu}\tr{V_\nu\rho_t V_\nu^\dag}} -\rho_t\right)
      \left(dN_\mu(t)-\Big( \thetab_\mu+\sum_{\nu}\etab_{\mu,\nu}\tr{V_\nu\rho_t V_\nu^\dag}\Big)\,dt\right)
\end{multline*}
driven by the  the Poisson processes $N_\mu(t)$ with
$
\langle dN_\mu(t)\rangle= \Big( \thetab_\mu+\sum_{\nu}\etab_{\mu,\nu}\tr{V_\nu\rho_t V_\nu^\dag}\Big)\,dt,
$
where  the detection  imperfections are  modeled through the parameters $\thetab_\mu\geq 0$ and $\etab_{\mu,\nu} \geq 0$ with  $\sum_{\mu} \etab_{\mu,\nu}  \leq 1$.

The proof of Theorem~\ref{thm:pws} is based on discrete-time approximations of  continuous-time SMEs. Such approximations have already been investigated  in \cite{attal2006repeated,gough2004stochastic,pellegrini2007existence,pellegrini2008existence}. They rely on indirect measurements, originally introduced in~\cite{Vladimir} and  well explained with suggestive physical systems   in~\cite{haroche-raimond:book06,wiseman-milburn:book}.
For   discrete-time SMEs, it is  proved in~\cite{rouchon2010fidelity, somaraju-et-al:acc2012} that the fidelity between the quantum state and its estimate is always a submartingale.  Theorem~\ref{thm:pws} is  obtained by passing to the limit from   discrete to   continuous time. The fidelity between the quantum state and its estimate  remains a submartingale. The obtained continuous-time SMEs are slightly more general than the ones usually encountered in the literature. Such SMEs could be of some interest to derive  quantum filters  taking into account a larger class of  incompleteness and errors in measurements and jump detections.

This paper is structured  in two  main sections. Section~\ref{sec:perfect} is devoted to Theorem~\ref{thm:continuous-time stability}, a restrictive version of Theorem~\ref{thm:pws} to  the diffusive SMEs with perfect measurements. For this simplified but representative case, the discrete-time approximation is presented and the passage to the continuous-time limit  is detailed during the proof of Proposition~\ref{pr:conv}.  Section~\ref{sec:imperfect} is devoted to Theorem~\ref{thm:pws} and fully exploits  the tools and methods  developed in Section~\ref{sec:perfect}.  We consider  SMEs driven simultaneously by Poisson and Wiener processes. We recall first the structure, described in~\cite{somaraju-et-al:acc2012} and based on a left stochastic matrix, of discrete-time SMEs associated to   imperfect measurements. We apply on the discrete-time approximations  such left stochastic matrix modeling of imperfections and errors. Then we take, thanks to Theorem~\ref{thm:pw},  the limit to get  the continuous-time SMEs and its associated quantum filters  with imperfections. Their structures are  more general than the usual ones encountered in the literature. This leads to Theorem~\ref{thm:pws} ensuring the stability of the obtained quantum filters.  Section~\ref{sec:conclusion} is a short conclusion proposing some connection with Petz characterization of monotone metrics on matrix spaces.

Some intermediate and partial results related to Theorems~\ref{thm:pws}  can be founded in~\cite{aminicdc,AminiPhD}.
%%%%%%%%%%%%%%%%%%%%%%%%%%%%%%%%%%%%%%%%%%%%%%%%%%%%%%%%%%%%%%%%%%%%%%%%%%%%%%%%%%%%%%%%
\section{Perfect measurements}~\label{sec:perfect}
Let us start this section by presenting the jump-diffusive SMEs  describing the evolutions of quantum  state $\rho_t$ and its estimate $\rho^e_t$.
\subsection{Continuous-time filters}~\label{sec:filterp}
We consider quantum systems of finite dimensions $1<N<\infty$. The state space of such a system is given by the set of density matrices
\begin{equation*}
\mathcal D:=\{\rho\in \mathbb C^{N\times N}|\quad \rho=\rho^\dag,\quad \tr{\rho}=1,\quad \rho\geq 0\}.
\end{equation*}
Formally a real quantum trajectory $\rho\in\mathcal D$ in Schr\"{o}dinger picture can be described by the following SME (cf.~\cite{belavkin1992quantum,Belavkin2007,Barchiellei1})
\begin{multline}\label{eq:rho}
    d\rho_t = \left(-i[H,\rho_t] + \sum_\xi \Lc_\xi(\rho_t) \right) dt +
    \sum_\nu \Lambda_\nu(\rho_t) dW_\nu(t) + \sum_\mu \Upsilon_\mu(\rho_t) \left(dN_\mu(t) -\tr{V_\mu\rho_t V_\mu^\dag}dt\right),
\end{multline}
where
\begin{itemize}
\item the notation $[A,B]$ refers to $AB-BA;$
\item $H=H^\dag$ is a Hermitian operator corresponding to the total Hamiltonian of the system;
\item $dN_\mu$ are the Poisson processes with $\mu\in\{1,\cdots,m_P\}$ and $dW_\nu$ are the Wiener processes with $\nu\in\{m_P+1,\cdots,m_P+m_W\}$;
\item The Lindblad superoperator $\Lc_\xi$ ($\xi\in\{1,\ldots,m_W+m_P\}$) is defined by
\begin{equation}~\label{eq:lindblad}
\Lc_\xi(\rho):= V_\xi\rho V_\xi^\dag - \half(V_\xi^\dag V_\xi\rho +\rho V_\xi^\dag V_\xi ),
\end{equation}
where $V_\xi$ is an arbitrary matrix which determines the measurement process (typically the coupling to the probe field for quantum optic systems);
\item The superoperators  $\Lambda_\nu$ and $\Upsilon_\mu$ are defined respectively by $$\Lambda_\nu(\rho):= V_\nu\rho + \rho V_\nu^\dag - \tr{(V_\nu+V_\nu^\dag)\rho} \rho \quad \textrm{and}\quad \Upsilon_\mu(\rho):=\frac{V_\mu \rho V_\mu^\dag}{\tr{V_\mu \rho V_\mu^\dag}} -\rho;$$
\item The measurement outcomes are $dN_\mu$ and $dy_\nu$ where
\begin{equation}~\label{eq:inn}
 N_\mu(t)-\int_0^t\tr{V_\mu\rho_s V_\mu^\dag}\,ds\quad\textrm{is a martingale and}\quad dy_\nu(t) = dW_\nu(t)+\tr{V_\nu\rho_t V_\nu^\dag}dt.
\end{equation}
\end{itemize}
All the developments remain valid when H and $V_\xi$ are deterministic time-varying matrices. For clarity sake, we do not recall below such possible time dependence.

\medskip

In this paper, the notation $\rho^e$ corresponds to  the estimate filter associated to the  filter $\rho.$ This estimate filter is  provided from the measurement outcomes $dy_\nu$ and $dN_\mu$ and depends on the real quantum trajectory $\rho$ via the measurement outcomes~\eqref{eq:inn}. It has  the following expression derived from~\eqref{eq:rho} where
$dW_\nu(t)$ is replaced by $dy_\nu(t) -\tr{V_\nu\rho^e_t V_\nu^\dag}dt$ (see e.g.,~\cite{Belavkin2007,Barchiellei1}):
\begin{multline}\label{eq:rhoe}
    d\rho^e_t = \left(-i[H,\rho^e_t] + \sum_\xi \Lc_\xi(\rho^e_t) \right) dt +
    \sum_\nu \Lambda_\nu(\rho^e_t)\left(dy_\nu(t) -\tr{V_\nu\rho^e_t V_\nu^\dag}dt\right) \\
    + \sum_\mu \Upsilon_\mu(\rho^e_t) \left(dN_\mu(t) -\tr{V_\mu\rho^e_t V_\mu^\dag}dt\right)
    .
\end{multline}
When $\rho^e_0\neq \rho_0$, $\rho^e_t$ and $\rho_t$ do not coincide in general. However, we will see that the fidelity between $\rho^e_t$ and $\rho_t$ is a sub-martingale.

\subsection{Discrete-time filters}
First let us briefly remind the model of quantum repeated indirect measurement approach. The physics underlying such approach is well explained in~\cite{haroche-raimond:book06,wiseman-milburn:book}. In~\cite{attal2006repeated,gough2004stochastic,pellegrini2010markov}, it was rigorously shown that  such discrete-time approximations associated to the real state converges to the continuous model described in Equation~\eqref{eq:rho}.

\subsubsection{Quantum repeated measurement approach}  We consider the setup of quantum repeated interaction of $\mathcal{H}$ (describing the Hilbert space of the system state) with an infinite chain $\bigotimes_k\mathcal{K}_k$ (describing the Hilbert space of the environment) with $\mathcal{K}_k=\mathcal{K}$ for all $k$. More precisely, the first copy $\mathcal{K}_1=\mathcal{K}$ interacts with $\mathcal{H}$ during a time $\delta$ and then disappear. Next, the second copy $\mathcal{K}_2$ comes to interact with $\mathcal{H}$ and so on. This setting is coupled with indirect measurement, that is, after each interaction between $\mathcal{H}$ and $\mathcal{K},$ a measurement of an observable of $\mathcal{K}$ is performed.

\medskip

Here, we assume that $\mathcal H=\CC^N$ and the environment (meter system) is composed of $m_P+m_W$ qubits, so $\mathcal K=(\CC^2)^{\otimes m_P+m_W}.$ Also, take $\ket\psi$ as the initial state of the system and $\GG$ as the initial state of the qubits' meter which is defined as the ground state of the meter system: all qubits in the ground state $\ket g$. As a result, the initial state of the system coupled to its environment is described by $\ket\psi\otimes \GG.$

\medskip

For $\delta=1/n$ with $n$ large, assume that the Schr\"{o}dinger  evolution between time $0$ to time $1/n$ is given by
\begin{equation}~\label{eq:htot}
\Htot = H \otimes I + \sqrt{n} \sum_{\xi} \left( i V_\xi \otimes \Spxi - i V_\xi^\dag \otimes \Smxi\right),
\end{equation}
where
\begin{itemize}
  \item $H$ is the Hamiltonian of the system used   in~\eqref{eq:rho} and the operators $V_\xi$ are those  appearing in the Lindblad superoperators~\eqref{eq:lindblad}
\item $\Spxi=(\ket{e}\bra{g})_{\xi}$ and $\Smxi=\Spxi^\dag=(\ket{g}\bra{e})_\xi$
 where the notation $(A)_\xi$ is an operator $\mathcal K$ defined by $\displaystyle(A)_\xi=\bigotimes_{1\leq i<\xi}I\otimes A\otimes\bigotimes _{\xi<i\leq m_W+m_P}I$. Note that $\Smxi \GG =0.$
\end{itemize}
In the sequel, symbol $\otimes$ will be remove for compact formulae.

%\red{We will also  use $\Gxi := \Spxi \GG $  corresponding to the state where all qubits are in the state $\ket{g}$ excepted qubit number $\xi$ which is  in the state $\ket{e}$.}

Take the observables $X_\nu=(\ket{g}\bra{e})_\nu+(\ket{e}\bra{g})_\nu$ for $\nu\in\{m_P+1,\cdots,m_W+m_P\}$ and $Z_\mu=(\ket{e}\bra{e})_\mu-(\ket{g}\bra{g})_\mu$ for $\mu\in\{1,\cdots,m_P\}$.  For $n$ large,  the measurement of all qubits at final time $t=1/n$ according to the observables  $X_\nu$ and $Z_\mu,$  yields an approximation   for the Wiener  and Poisson processes, respectively.  This results from a  development  versus $1/n$ of the measurement operators associated to the associated  discrete-time  stochastic evolution.

We have
$ e^{-i\Htot/n}= \Id - \frac{i}{n} \Htot - \frac{1}{2n^2} \Htot^2 +\mathcal O(1/n^{3/2}),$ where $\Id$ is the identity operator. Now replacing $\Htot$ by its expression given in~\eqref{eq:htot}, we find
$$
e^{-i\Htot/n} \approx  \Id + \tfrac{1}{\sqrt{n}} \sum_\xi \left( V_\xi \Spxi-  V_\xi^\dag \Smxi\right)  - \tfrac{i}{n} H -
\tfrac{1}{2n} \sum_\xi \left(V_\xi^\dag V_\xi \Pgxi + V_\xi V_\xi^\dag \Pexi   \right),
$$
where $\approx$ means up to $\mathcal O(1/n^{3/2})$ terms. Here,  $\Pgxi=(\ket{g}\bra{g})_\xi$ and $\Pexi=(\ket{e}\bra{e})_\xi.$ Note that
$$\Pgxi\ket G=\ket G\quad\textrm{and}\quad\Pexi\ket G=0.$$
Thus the system coupled to its environment  evolves as follows
$$
    e^{-i\Htot/n} \ket\psi \otimes \GG \approx
    \left(\Id - \tfrac{i}{n} H - \tfrac{1}{2n} \sum_\xi V_\xi^\dag V_\xi \right)\ket\psi \otimes\GG
    + \tfrac{1}{\sqrt{n}} \sum_\xi  V_\xi\ket\psi \otimes \Gxi
    .
$$
Let us consider the qubits measurements. The measurement outcomes for  $X_\nu$ are stored  in $x_{\nu-m_P}\in\{-1,1\}$, $x=(x_{\nu-m_P})\in \{-1,1\}^{m_W}$, and those of $Z_\mu$ in
$z_\mu\in\{0,1\}$, $z=(z_\mu)\in\{0,1\}^{m_P}$,  as follows (all these observables commute):
\begin{itemize}
  \item  if during the   measure  of  $X_\nu$, the corresponding qubit collapses to $(\ket g+\ket e)/\sqrt{2} $ (resp. $(\ket g - \ket e)/\sqrt{2}$) set $x_{\nu-m_P}=+1$ (resp. $x_{\nu-m_P}=-1$).
  \item if during the measure of $Z_\mu$, the corresponding qubit collapses to $\ket g$ (resp.$\ket e$), set $z_\mu=0$ (resp. $z_\mu=1$).
\end{itemize}
The probability to get, for two different  $\mu$ and $\mu'$, $z_\mu=z_{\mu'}=1$, is in order of $\mathcal O(1/n^{3/2})$. Thus for $z,$ we need only consider  the following  cases:  either all $z_\mu$ are equal to~$0$  or  only a single one is equal to~$1$, the other ones being $0$.

\medskip

Consider the measurement outcomes $z=(z_\mu)$ and $x=(x_{\nu-m_P})$.  The associated  wave packet collapse of $e^{-i\Htot/n} \ket\psi \otimes \GG $  yields the un-normalized state $M_{x,z} \ket\psi \otimes \ket{x,z},$ where the measurement operator is denoted by $M_{x,z}$ and  where $\ket{x,z}$  is the normalized state of the qubits  characterized  by $X_\nu \ket{x,z} = x_{\nu-m_P} \ket{x,z}$ and $Z_\mu \ket{x,z}= (2z_\mu-1)\ket{x,z}$. We have to consider two situations: either $\sum z_\mu=0$ denoted by $z=0$ or $\sum_\mu z_\mu=1$.
When  $z_\mu=0$ for all $\mu$, some  computations yield
\begin{equation}\label{eq:Mx0}
    M_{x,0} \approx \left(\prod_\nu \tfrac{1}{\sqrt{2}}\right)
    \left( \Id - \tfrac{i}{n} H - \tfrac{1}{2n} \sum_\xi V_\xi^\dag V_\xi + \tfrac{1}{\sqrt{n}} \sum_\nu  x_{\nu-m_P} V_\nu \right)
    .
\end{equation}
When   $z_\mu=1$ and $z_{\mu^\prime}=0$ for all $\mu^\prime\neq \mu$,  similar  computations give
\begin{equation}\label{eq:Mx1}
    M_{x,\mu} \approx    \left(\prod_\nu \tfrac{1}{\sqrt{2}}\right) \frac{V_\mu}{\sqrt{n}}
\end{equation}
where we have denoted $M_{x,z}$ by $M_{x,\mu}$ for such $z$.

\medskip

For an arbitrary state $\rho$ at $t=0$, not necessarily pure as $\rho=\ket\psi\bra\psi$, the state $\rho_1$ at time $t=1/n$ is given by
\begin{equation}~\label{eq:truedyn}
\rho_1 = \frac{M_{x,z} \rho M_{x,z}^\dag}{p_{x,z}(\rho)},
\end{equation}
which happens with probability $p_{x,z}(\rho)= \tr{M_{x,z}\rho M_{x,z}^\dag}$.

The expression of $\rho_1$ is obtained by neglecting the terms of orders strictly greater than $1$ versus $1/n.$ Some usual calculations yield
\begin{multline}~\label{eq:MMxO}
    M_{x,0} \rho M_{x,0}^\dag = \\
    c_W \left(\rho + \tfrac{1}{\sqrt{n}} \sum_\nu x_{\nu-m_P} (  V_\nu \rho+ \rho V_\nu^\dag) +
    \tfrac{1}{n} \Big(-i[H,\rho] - \sum_\mu \tfrac{\{\rho,V_\mu^\dag V_\mu\}}{2}+ \sum_\nu \Lc_\nu(\rho) \Big) \right)
\end{multline}
with probability
\begin{equation}\label{eq:PxO}
    p_{x,0}(\rho) =
    c_W \left(1 + \tfrac{1}{\sqrt{n}} \sum_\nu x_{\nu-m_P} \tr{  (V_\nu+V_\nu^\dag)\rho}-
    \tfrac{1}{n} \sum_\mu \tr{\rho V_\mu^\dag V_\mu} \right)
    .
\end{equation}
The notation $\{A,B\}$ used in~\eqref{eq:MMxO} corresponds to $AB+BA$ and $c_W=(\half)^{m_W}.$

\medskip

Similarly, we have
\begin{equation}\label{eq:MMPxmu}
     M_{x,\mu} \rho M_{x,\mu}^\dag = \tfrac{c_W}{n} V_\mu \rho V_\mu^\dag, \quad\textrm{with probability}\quad
     p_{x,\mu}(\rho) = \tfrac{c_W}{n} \tr{V_\mu \rho V_\mu^\dag}
     .
\end{equation}
Now we can obtain the asymptotic description of the transition for all possible observations. Indeed, the expression of the whole quantum trajectory $\rho_k$ can be obtained by replacing $\rho_1$ by $\rho_{k+1}$ and $\rho$ by $\rho_k$ in Equation~\eqref{eq:truedyn}.

As a conclusion, we find
\begin{align}~\label{eq:discrete}
\rho_{k+1}=
\frac{M_{x,\mu}\rho_k M_{x,\mu}^\dag}{\tr{M_{x,\mu}\rho_k M_{x,\mu}^\dag}}\quad \textrm{with probability}\quad p_{x,\mu}(\rho_k)\\
\end{align}
where  $p_{x,\mu}(\rho_k)$ is given by~(\ref{eq:MMPxmu},\ref{eq:PxO}).
\subsubsection{Stability with respect to initial condition}
Consider the Markov chain~\eqref{eq:discrete}. Assume that we do not know precisely the initial state $\rho_0$ and we have at our disposal an estimate $\rho^e_0$. Assume also that we know the measurement result at step $k$ $(x_k,\mu_k)\in \{-1,1\}^{m_W}\times \{0,1\}$. It is then natural to consider the following recursive update of our estimation $\rho^e_{k+1}$ using the knowledge of measurement result at step $k$ and the previous estimate $\rho^e_k$ (see e.g.,~\cite{wiseman-milburn:book}):
\begin{align}~\label{eq:ediscrete}
\rho^e_{k+1}=
 \frac{M_{x_k,\mu_k}\rho^e_k M_{x_k,\mu_k}^\dag}{\tr{M_{x_k,\mu_k}\rho^e_k M_{x_k,\mu_k}^\dag}}
\end{align}
Note that the probability  $p_{x_k,\mu_k}(\rho_k)$ given by (\ref{eq:PxO},\ref{eq:MMPxmu})  to get $(x_k,\mu_k)$  depends on the hidden state $\rho_k$ and not on $\rho^e_k.$
\begin{rk}~\label{rk:per}
 \rm Let us stress that the above description is not always valid. Indeed, the normalization $\tr{M_{x_k,\mu_k}\rho^e_k M_{x_k,\mu_k}^\dag}$  can vanish and the formula \eqref{eq:ediscrete} is then not defined. This problem does not appear when describing the true evolution since the normalization describing the true state corresponds to the probability of apparition (then if this vanishes this means that the corresponding state can not appear). This problem of non-definition for the discrete-time estimate filter is related to the problem underlined in the definition of \eqref{eq:rhoe}. In general, this question has been taken into account in~\cite{rouchon2010fidelity,somaraju-et-al:acc2012} for the discrete-time filter. In our context focusing on asymptotic evolution, such a problem will not appear.
\end{rk}
\begin{thm}[\cite{rouchon2010fidelity}]~\label{thm:discrete-time stability}
\rm  Consider any arbitrary Markov chain $(\rho_k,\rho^e_k)$ satisfying respectively Equations~\eqref{eq:discrete} and~\eqref{eq:ediscrete}:
 \begin{align*}
(\rho_{k+1},\rho^e_{k+1})=
\left(\tfrac{M_{x,\mu}\rho_k M_{x,\mu}^\dag}{\tr{M_{x,\mu}\rho_k M_{x,\mu}^\dag}},\tfrac{M_{x,\mu}\rho^e_k M_{x,\mu}^\dag}{\tr{M_{x,\mu}\rho^e_k M_{x,\mu}^\dag}}\right)\quad \textrm{with probability}\quad p_{x,\mu}(\rho_k)
\end{align*}
where $p_{x,\mu}(\rho_k)$ is  given by~\eqref{eq:MMPxmu} and~\eqref{eq:PxO}.
 Then the fidelity $(F(\rho_k,\rho^e_k))$ defined in Equation~\eqref{eq:fidelity} is a $(\mathcal F_k)$ submartingale where $\mathcal F_k=\sigma\{(\rho_l,\rho^e_l)\vert l\leq k)\}$. In particular, we have
\begin{equation*}
\EE{F(\rho_{l},\rho^e_{l})|\mathcal F_k}=\EE{F(\rho_{l},\rho^e_{l})|(\rho_k,\rho^e_k)}\geq F(\rho_k,\rho^e_k),
\end{equation*}
 for all $l>k$.
\end{thm}
\subsection{Stability of continuous-time filters}~\label{sec:mr}
From the description of the quantum trajectory $(\rho_k)$ and its associated quantum filter $(\rho^e_k)$ with the time parameter $n,$ we can define the associated continuous-time stochastic processes denoted by $(\rho_n(t))$ and $(\rho^e_n(t))$ with
\begin{equation*}
\rho_n(t)=\rho_{[nt]},\quad \rho^e_n(t)=\rho^e_{[nt]}.
\end{equation*}
It is clear that $(\rho_n,\rho^e_n)$ is a Markov process as $(\rho_k,\rho^e_k)$ is a Markov chain.

Before announcing the main result of this section, first let us show the following proposition.
\begin{pr}~\label{pr:conv}
 The Markov process $(\rho_n(t),\rho^e_n(t))$ with $\rho_n$ and $\rho^e_n$ satisfying respectively Equations~\eqref{eq:discrete} and~\eqref{eq:ediscrete} converges in distribution, when $n$ goes to infinity, to the Markov process $(\rho_t,\rho^e_t)$ with $\rho_t$ and $\rho^e_t$ satisfying respectively Equations~\eqref{eq:rho} and~\eqref{eq:rhoe}.
\end{pr}
\begin{proof}
The approach to prove this proposition is usual in probability theory: we show the convergence of the generator associated to the Markov process $(\rho_n,\rho^e_n)$ towards the one associated to the Markov process $(\rho_t,\rho^e_t).$ Also, we need to prove the tightness property of the sequence $(\rho_n,\rho^e_n)$.

The tightness of the sequence $(\rho_n(t),\rho^e_n(t))$ is guaranteed if, for any $T>0$,   exists $M>0$ such that for all $0\leq t_1 \leq t_2\leq T$ and  $\forall t\in[t_1,t_2]$,
\begin{equation*}
\mathbb{E}\left[\Vert(\rho_n(t_2),\rho^e_n(t_2))-(\rho_n(t),\rho^e_n(t))\Vert^2\Vert(\rho_n(t),\rho^e_n(t))-(\rho_n(t_1),\rho^e_n(t_1))\Vert^2\right]\leq M(t_2-t_1)^2
.
\end{equation*}
We do not develop the arguments for showing the tightness property. They can be obtained directly with the same arguments as~\cite{pellegrini2010markov}.

Thus here, we just focus  on the convergence of the generators. Take $\Ac_n$ and $\Ac$ as the generators associated respectively to the sequences $(\rho_n,\rho^e_n)$ and $(\rho,\rho^e).$
We have to prove that, for any $C^2-$real valued function $(\rho,\rho^e) \mapsto f(\rho,\rho^e)$,
$$
\lim_{n\rightarrow +\infty} \sup_{(\rho,\rho^e)\in\mathcal D^2} |\Ac_n f(\rho,\rho^e) - \Ac f(\rho,\rho^e)| =0.
$$
We have the following expression for the generator $\Ac f(\rho,\rho^e)$
\begin{multline}\label{eq:A}
    \Ac f(\rho,\rho^e)= D_{(\rho,\rho^e)} f.\left(-i[H,\rho] + \sum_\xi \Lc_\xi(\rho), -i[H,\rho^e] + \sum_\xi \Lc_\xi(\rho^e)+K(\rho,\rho^e)\right)
     \\
    + \tfrac{1}{2} \sum_\nu D^2_{(\rho,\rho^e)} f.\left(\Lambda_\nu(\rho), \Lambda_\nu(\rho^e);\Lambda_\nu(\rho), \Lambda_\nu(\rho^e) \right)
    \\
    + \sum_\mu \tr{V_\mu \rho V_\mu^\dag} \left(f\left(\MK{V_\mu \rho V_\mu^\dag}, \MK{V_\mu \rho^e V_\mu^\dag}\right)-f(\rho,\rho^e)-D_{(\rho,\rho^e)}f.\left(\Upsilon_\mu(\rho),\Upsilon_\mu(\rho^e)\right)\right),
\end{multline}
where $$K(\rho,\rho^e):=\sum_\mu \Upsilon_\mu(\rho^e)\Big(\tr{V_\mu\rho V_\mu^\dag}-\tr{V_\mu\rho^e V_\mu^\dag}\Big)+\sum_\nu \Lambda_\nu(\rho^e)\Big(\tr{V_\nu\rho V_\nu^\dag}-\tr{V_\nu\rho^e V_\nu^\dag}\Big)$$
Now let us calculate the expression of $\Ac_n.$ We have, up to $\mathcal O(n^{-1/2})$ terms,
\begin{multline*}
    \Ac_n f(\rho,\rho^e)\approx
    \sum_x  n p_{x,0}(\rho)\left( f\left( \MK{M_{x,0}\rho M_{x,0}^\dag},\MK{M_{x,0}\rho^e M_{x,0}^\dag}\right) -f(\rho,\rho^e)\right)\\
    + \sum_{x,\mu} n p_{x,\mu}(\rho) \left(f\left(\MK{M_{x,\mu} \rho M_{x,\mu}^\dag }, \MK{M_{x,\mu} \rho^e M_{x,\mu}^\dag }\right)-f(\rho,\rho^e)\right)
.
\end{multline*}
Since $\frac{M_{x,0}\rho M_{x,0}^\dag}{p_{x,0}(\rho)}= \rho + \mathcal O(n^{-1/2})$ and also, by using Taylor's formula, we have
\begin{multline*}
f\left( \frac{M_{x,0}\rho M_{x,0}^\dag}{p_{x,0}(\rho)},\frac{M_{x,0}\rho^e M_{x,0}^\dag}{p_{x,0}(\rho^e)}\right)  - f(\rho,\rho^e)= D_{(\rho,\rho^e)}f\cdot \left( \frac{M_{x,0}\rho M_{x,0}^\dag}{p_{x,0}(\rho)} -\rho, \frac{M_{x,0}\rho^e M_{x,0}^\dag}{p_{x,0}(\rho^e)} -\rho^e\right) \\
+ \tfrac{1}{2} D^2_{(\rho,\rho^e)} f\cdot \left( \frac{M_{x,0}\rho M_{x,0}^\dag}{p_{x,0}(\rho)} -\rho, \frac{M_{x,0}\rho^e M_{x,0}^\dag}{p_{x,0}(\rho^e)}-\rho^e; \frac{M_{x,0}\rho M_{x,0}^\dag}{p_{x,0}(\rho)} -\rho, \frac{M_{x,0}\rho^e M_{x,0}^\dag}{p_{x,0}(\rho^e)}-\rho^e \right)+\mathcal O(n^{-3/2}).
\end{multline*}
Consequently,
{\small\begin{align}~\label{eq:gen}
    &n \sum_x  p_{x,0}(\rho) \left(f\left( \frac{M_{x,0}\rho M_{x,0}^\dag}{p_{x,0}(\rho)},\frac{M_{x,0}\rho^e M_{x,0}^\dag}{p_{x,0}(\rho^e)}\right) - f(\rho,\rho^e)\right)\nonumber\\
    &= n D_{(\rho,\rho^e)} f\cdot \left( \sum_x M_{x,0}\rho M_{x,0}^\dag- p_{x,0}(\rho)\rho,\sum_x \Big(M_{x,0}\rho^e M_{x,0}^\dag- p_{x,0}(\rho^e)\rho^e\Big)\frac{p_{x,0(\rho)}}{p_{x,0(\rho^e)}}\right)
   \nonumber \\
    &+ \tfrac{1}{2} \sum_x n p_{x,0}(\rho)  D^2_{(\rho,\rho^e)} f\cdot \left( \frac{M_{x,0}\rho M_{x,0}^\dag}{p_{x,0}(\rho)} -\rho, \frac{M_{x,0}\rho^e M_{x,0}^\dag}{p_{x,0}(\rho^e)}-\rho^e;\frac{M_{x,0}\rho M_{x,0}^\dag}{p_{x,0}(\rho)} -\rho, \frac{M_{x,0}\rho^e M_{x,0}^\dag}{p_{x,0}(\rho^e)}-\rho^e \right)+ \mathcal O(n^{-1/2}),
\end{align}}
where according to~\eqref{eq:MMxO} and~\eqref{eq:PxO}, up to $\mathcal O(n^{-1/2})$ terms, we have
\begin{align}~\label{eq:auxO}
n \sum_x \left(M_{x,0}\rho M_{x,0}^\dag- p_{x,0}(\rho)\rho\right) &\approx
-i[H,\rho] + \sum_\mu \left(\tr{\rho V_\mu^\dag V_\mu} \rho - \tfrac{\{\rho,V_\mu^\dag V_\mu\}}{2}\right)+ \sum_\nu \Lc_\nu(\rho)\nonumber\\
&=-i[H,\rho] +\sum_\xi \Lc_\xi(\rho)-\sum_\mu\Upsilon_{\mu}(\rho)\tr{V_\mu\rho V_\mu^\dag},
\end{align}
since, for any $\nu$,  $\sum_x x_{\nu-m_P}\Lambda_\nu(\rho)=0$ and $\sum_x c_W=1$.
Also, we have
\begin{multline*}
M_{x,0}\rho^e M_{x,0}^\dag-p_{x,0}(\rho^e)\rho^e\\
=c_W\left(\frac{1}{\sqrt{n}}\sum_\nu x_{\nu-m_P}\Lambda_\nu(\rho^e)+\frac{1}{n}\Big(-i[H,\rho^e]+\sum_\xi\mathcal L_\xi(\rho^e)-\sum_\mu\Upsilon_\mu(\rho^e)\tr{V_\mu\rho^e V_\mu^\dag}\Big)\right).
\end{multline*}
Therefore, we find
\begin{multline}~\label{eq:fterme}
n\sum_x \Big(M_{x,0}\rho^e M_{x,0}^\dag-p_{x,0}(\rho^e)\rho^e\Big)\frac{p_{x,0}(\rho)}{p_{x,0}(\rho^e)}\\
=-i[H,\rho^e]+\sum_\xi\mathcal L_\xi(\rho^e)-\sum_\mu\Upsilon_\mu(\rho^e)\tr{V_\mu\rho^e V_\mu^\dag}\\
+\sum_\nu \Lambda_\nu(\rho^e)\Big(\tr{V_\nu\rho V_\nu^\dag}-\tr{V_\nu\rho^e V_\nu^\dag}\Big)+\mathcal O(n^{-\half})
\end{multline}
since $\frac{p_{x,0}(\rho)}{p_{x,0}(\rho^e)}=1+\frac{1}{\sqrt{n}}\sum_\nu x_{\nu-m_P}\Big(\tr{(V_\nu+V_\nu^\dag)\rho}-\tr{(V_\nu+V_\nu^\dag)\rho^e}\Big)+\mathcal O(n^{-1/2})$.
 Moreover, we have used the fact that
\begin{multline*}
\sum_x \sum_\nu c_W x_{\nu-m_P} \Lambda_\nu(\rho^e)\sum_{\nu'}x_{\nu'-m_P}\left(\tr{(V_{\nu'}+V_{\nu'}^\dag)\rho}-\tr{(V_{\nu'}+V_{\nu'}^\dag)\rho^e}\right)\\
=\sum_\nu\Lambda_\nu(\rho^e)\left(\tr{(V_\nu+V_\nu^\dag)\rho}-\tr{(V_\nu+V_\nu^\dag)\rho^e}\right),
\end{multline*}
since $x_{\nu-m_P}^2=1, \sum_x c_W x_{\nu-m_P}^2=1$ and $$\forall \nu\neq \nu':\quad\sum_x c_W x_{\nu-m_P} x_{\nu'-m_P}\Lambda_\nu(\rho^e)\left(\tr{(V_{\nu'}+V_{\nu'}^\dag)\rho}-\tr{(V_{\nu'}+V_{\nu'}^\dag)\rho^e}\right)=0.$$
Equation~\eqref{eq:fterme} can be rewritten as follows
\begin{multline*}
n\sum_x \Big(M_{x,0}\rho^e M_{x,0}^\dag-p_{x,0}(\rho^e)\rho^e\Big)\frac{p_{x,0}(\rho)}{p_{x,0}(\rho^e)}\\
=-i[H,\rho^e]+\sum_\xi\mathcal L_\xi(\rho^e)+K(\rho,\rho^e)-\sum_\mu\Upsilon_\mu(\rho^e)\tr{V_\mu\rho V_\mu^\dag}+\mathcal O(n^{-\half}).
\end{multline*}
As a result, the first righthand side term in Equation~\eqref{eq:gen} can be written as follows
\begin{multline*}
n D_{(\rho,\rho^e)} f\cdot \left( \sum_x M_{x,0}\rho M_{x,0}^\dag- p_{x,0}(\rho)\rho,\sum_x \Big(M_{x,0}\rho^e M_{x,0}^\dag- p_{x,0}(\rho^e)\rho^e\Big)\frac{p_{x,0(\rho)}}{p_{x,0(\rho^e)}}\right)\\
=D_{(\rho,\rho^e)} f.\left(-i[H,\rho] + \sum_\xi \Lc_\xi(\rho), -i[H,\rho^e] + \sum_\xi \Lc_\xi(\rho^e)+K(\rho,\rho^e)\right)\\
-\sum_\mu\tr{V_\mu\rho V_\mu^\dag}D_{(\rho,\rho^e)}f\cdot\Big(\Upsilon_\mu(\rho),\Upsilon_\mu(\rho^e)\Big)+\mathcal O(n^{-1/2}).
\end{multline*}
Now let us calculate the second righthand side term in Equation~\eqref{eq:gen}. To get the zero order terms of $\sum_x n p_{x,0}(\rho)D^2_{(\rho,\rho^e)} f\cdot \left( \frac{M_{x,0}\rho M_{x,0}^\dag}{p_{x,0}(\rho)} -\rho, \frac{M_{x,0}\rho^e M_{x,0}^\dag}{p_{x,0}(\rho^e)}-\rho^e;\frac{M_{x,0}\rho M_{x,0}^\dag}{p_{x,0}(\rho)} -\rho, \frac{M_{x,0}\rho^e M_{x,0}^\dag}{p_{x,0}(\rho^e)}-\rho^e \right),$ we just need to combine the terms of order $n^{-1/2}$  in $\frac{M_{x,0}\rho M_{x,0}^\dag}{p_{x,0}(\rho)} -\rho$ and in $M_{x,0}\rho M_{x,0}^\dag- p_{x,0}(\rho)\rho$. Since
$$
M_{x,0}\rho M_{x,0}^\dag- p_{x,0}(\rho)\rho = \tfrac{c_W }{\sqrt n}
\left( \sum_\nu x_{\nu-m_P} \Lambda_\nu(\rho)  \right) + \mathcal O(n^{-1/2})
$$
and
$$
  \frac{M_{x,0}\rho M_{x,0}^\dag}{p_{x,0}(\rho)}-\rho  \approx
  \frac{ \rho + \tfrac{1}{\sqrt{n}} \sum_\nu x_{\nu-m_P} (  V_\nu \rho+ \rho V_\nu^\dag)}
       {1 + \tfrac{1}{\sqrt{n}} \sum_\nu x_{\nu-m_P} \tr{(V_\nu+ V_\nu^\dag)\rho}}-\rho
    \approx  \tfrac{1}{\sqrt{n}}\left( \sum_\nu x_{\nu-m_P} \Lambda_\nu(\rho)  \right) +\mathcal O(n^{-1/2}),
$$
we get, up to $\mathcal O(n^{-1/2})$ terms
\begin{multline*}
    \sum_x n p_{x,0}(\rho) D^2_{(\rho,\rho^e)} f\cdot\left( \frac{M_{x,0}\rho M_{x,0}^\dag}{p_{x,0}(\rho)} -\rho, \frac{M_{x,0}\rho^e M_{x,0}^\dag}{p_{x,0}(\rho^e)}-\rho^e;\frac{M_{x,0}\rho M_{x,0}^\dag}{p_{x,0}(\rho)} -\rho, \frac{M_{x,0}\rho^e M_{x,0}^\dag}{p_{x,0}(\rho^e)}-\rho^e\right)
    \approx \\
    \sum_{x,\nu,\nu'} c_W x_{\nu-m_P} x_{\nu'} D^2_{(\rho,\rho^e)} f\cdot \left(\Lambda_\nu(\rho), \Lambda_{\nu}(\rho^e); \Lambda_{\nu'}(\rho), \Lambda_{\nu'}(\rho^e)\right)
    = \sum_\nu D^2_{(\rho,\rho^e)} f\cdot\left( \Lambda_\nu(\rho), \Lambda_{\nu}(\rho^e);\Lambda_\nu(\rho), \Lambda_{\nu}(\rho^e) \right),
\end{multline*}
where for the above equality, we have used the following facts $$x_{\nu-m_P}^2=1, \quad \sum_x c_W=1\quad\textrm{and}\quad\forall \nu\neq\nu':\quad\sum_x c_W x_{\nu-m_P} x_{\nu'-m_P} D^2_\rho f(\rho)\cdot \left( \Lambda_\nu(\rho), \Lambda_{\nu'}(\rho) \right)=0.$$
Thus, we find
\begin{multline}~\label{eq:genf}
 \sum_x  p_{x,0}(\rho)\left( f\left( \frac{M_{x,0}\rho M_{x,0}^\dag}{p_{x,0}(\rho)},\frac{M_{x,0}\rho^e M_{x,0}^\dag}{p_{x,0}(\rho^e)} \right) -f(\rho,\rho^e)\right)
    \\
    =D_{(\rho,\rho^e)} f.\left(-i[H,\rho] + \sum_\xi \Lc_\xi(\rho), -i[H,\rho^e] + \sum_\xi \Lc_\xi(\rho^e)+K(\rho,\rho^e)\right)\\
-\sum_\mu\tr{V_\mu\rho V_\mu^\dag}D_{(\rho,\rho^e)}f\cdot\Big(\Upsilon_\mu(\rho),\Upsilon_\mu(\rho^e)\Big).
\\
+\tfrac{1}{2}\sum_\nu  D^2_\rho f(\rho) \cdot \left(\Lambda_\nu(\rho),\Lambda_\nu(\rho^e);\Lambda_\nu(\rho),\Lambda_\nu(\rho^e)\right)+\mathcal O(n^{-1/2})
    .
\end{multline}
According to~\eqref{eq:MMPxmu}, we have also
\begin{multline}~\label{eq:gens}
\sum_{x,\mu} np_{x,\mu}(\rho) \left(f\left(\MK{M_{x,\mu}\rho M_{x,\mu}^\dag},\MK{M_{x,\mu}\rho^e M_{x,\mu}^\dag} \right)-f(\rho,\rho^e)\right)
\\=
\sum_\mu \tr{V_\mu \rho V_\mu^\dag} \left( f\left(\MK{V_\mu \rho V_\mu^\dag},\MK{V_\mu \rho^e V_\mu^\dag}\right) - f(\rho,\rho^e)\right).
\end{multline}
Finally, by Equations~\eqref{eq:genf} and~\eqref{eq:gens}, we find that $$\Ac_n f(\rho,\rho^e) = \Ac f(\rho,\rho^e) + \mathcal O(n^{-1/2}),$$ which finishes the proof of Proposition~\ref{pr:conv}.
\end{proof}
 Now we are in the state to announce the main result of this section.
\begin{thm}~\label{thm:continuous-time stability}
 Consider the Markov process $(\rho_t,\rho^e_t)$ satisfying respectively Equations~\eqref{eq:rho} and~\eqref{eq:rhoe}. Then the fidelity  $(F(\rho_t,\rho^e_t))$ defined in Equation~\eqref{eq:fidelity} is a $(\mathcal F_t)-$submartingale, where $\mathcal F_t=\sigma\{(\rho_\tau,\rho^e_\tau)|\tau\leq t\}$. In particular, we have,  $$\EE{F(\rho_\tau,\rho^e_\tau)|\mathcal F_t}=\EE{F(\rho_\tau,\rho^e_\tau)|(\rho_t,\rho^e_t)}\geq F(\rho_t,\rho^e_t),$$ for all $\tau\geq t.$
\end{thm}
\begin{proof}
By Theorem~\ref{thm:discrete-time stability}, we know that the fidelity $F(\rho_n,\rho^e_n)$ is a submartingale with respect to the natural filtration of $(\rho_k,\rho^e_k)$.
  In terms of $(\rho_n(t),\rho^e_n(t)$, it follows that
\begin{equation}\label{eq:submartingale}
\EE{F(\rho_n(\tau),\rho^e_n(\tau))|(\rho_n(t),\rho^e_n(t)) }\geq {F(\rho_n(t),\rho^e_n(t))},
\end{equation}
for all $\tau\geq t.$  In Proposition~\ref{pr:conv}, we showed that $(\rho_n(t),\rho^e_n(t))$  converges in distribution to $(\rho_t,\rho^e_t)$.
This implies  that for any continuous real function $f(\rho,\rho^e)$ we have
$$
\forall \tau \geq t, \quad \lim_{n\rightarrow+\infty} \EE{f(\rho_n(\tau),\rho^e_n(\tau))|(\rho_n(t),\rho_n^e(t)}  =  \EE{f(\rho_\tau,\rho^e_\tau)|(\rho_t,\rho_t^e)}
.
$$
As $F$  is continuous, the limit of~\eqref{eq:submartingale} for $n$ tending to $\infty$ yields
$\EE{F(\rho_\tau,\rho^e_\tau)| (\rho_t,\rho_t^e)}\geq F(\rho_t,\rho^e_t)$
for all $\tau\geq t.$
Since  $(\rho_t,\rho^e_t)$ is a Markov process, the result follows.
\end{proof}
%%%%%%%%%%%%%%%%%%%%%%%%%%%%%%%%%%%%%%%%%%%%%%%%%%%%%%%%%%%%%%%%%%%%%%%%%%%%%%%%%%%%%%%%%
\section{Imperfect measurements}~\label{sec:imperfect}
 By imperfect measurements, we mean both unread measurements performed by the environment (decoherence) and active measurements performed by non-ideal detectors. Starting from the discrete-time case  experimentally used in~\cite{sayrin-et-al:nature2011} and detailed  in~\cite{somaraju-et-al:acc2012},  we derive the continuous-time optimal filters driven either by Poisson, Wiener processes or both of them. We obtain the SMEs and  their associated quantum  filters.

\subsection{Discrete-time  filters}~\label{sec:discretetime}
The presentation here is very much inspired from~\cite{somaraju-et-al:acc2012} and adapted to~\eqref{eq:discrete}  for $n$ large. The  jump-events labelled by $(x,z)$  are of type $(x,\mu)$ with $\mu \in\{0,1,\cdots,m_P\}$. The effectively  measured events are  labelled by $s \in\{1,\ldots,m\}$, with $m$ denoting the number of distinct  experimental detector outcomes. Suppose that we know the correlation between the  jump-events $(x,\mu)$ and  the experimental  detection $s$. These correlations  are modeled here  by classical
probabilities through a stochastic matrix $\eta$:  $\eta^n_{s, (x,\mu)}$ that gives the probability of experimental  detection  $s$ knowing that the effective jump-event   is $(x,\mu).$
 Since $\eta^n_{s,(x,\mu)}\geq 0$ and for each $(x,\mu),$ $\sum_{s=1}^m\eta^n_{s,(x,\mu)}=1,$ the matrix $\eta^n=(\eta^n_{s,(x,\mu)})$ is a left stochastic matrix.  We assume the following asymptotic for  $\eta^n$,
 \begin{equation}\label{eq:etan}
   \eta^n= \etaI  +\frac{\etaT}{n}+ \mathcal O(n^{-2}),
 \end{equation}
where $\etaI$ is a left stochastic matrix.

We still denote by $\rho$ the state associated to these experimental  detections which is the best estimation of the system state knowing the initial state and all previous experimental  detections. Following~\cite{somaraju-et-al:acc2012}, it obeys to the following Markov process:
 \begin{equation}~\label{eq:discrete-imp}
\rho_{k+1}=\frac{\Mc_s(\rho_k)}{\tr{\Mc_s(\rho_k)}}\quad \textrm{ with probability } p_s(\rho_k)=\tr{\Mc_s(\rho_k)}
\end{equation}
 where  $\Mc_s(\rho)\triangleq \sum_{x,\mu}\eta^n_{s,(x,\mu)} M_{x,\mu}\rho M_{x,\mu}^\dag,$   $M_{x,0}$ and $M_{x,\mu}$ are given respectively by Equations~\eqref{eq:Mx0} and~\eqref{eq:Mx1} depending also on $n$.

 Suppose that the initial state of dynamics~\eqref{eq:discrete-imp} is not well known. Let $\rho^e_0$ be an arbitrary initial estimate, the estimate discrete-time filter satisfies the following dynamics
\begin{equation}~\label{eq:ediscrete-imp}
\rho^e_{k+1}=\frac{\Mc_{s_k}(\rho^e_k)}{\tr{\Mc_{s_k}(\rho^e_k)}}
\end{equation}
where  $s_k$ corresponds to the experimental detection at time-step $k$.
\begin{rk}
\rm For the same reason given in Remark~\ref{rk:per}, the above description is not always valid. Since the normalization $\tr{\Mc_s(\rho^e_k)}$ can vanish and the formula~\eqref{eq:ediscrete-imp} is then not well defined. Such a problem will not be appeared when $\rho^e_k$ is full rank. (When $\tr{\Mc_s(\rho^e_k)}=0,$ we can still define the value of $\rho^e_{k+1},$ see more details in~\cite{somaraju-et-al:acc2012}.) Again, as we consider the asymptotic evolution, such a problem will not appear.
\end{rk}
We now state a theorem ensuring the stability of such estimation procedure whatever the initial state $\rho^e_0$ is.
\begin{thm}[\cite{rouchon2010fidelity,somaraju-et-al:acc2012}]~\label{thm:discrete-imp}
\rm Consider the  Markov chain $(\rho_k,\rho^e_k)$ satisfying~\eqref{eq:discrete-imp} and~\eqref{eq:ediscrete-imp}:
$$
(\rho_{k+1},\rho_{k+1}^e) =\left(\tfrac{\Mc_s(\rho_k)}{\tr{\Mc_s(\rho_k)}}, \tfrac{\Mc_s(\rho^e_k)}{\tr{\Mc_s(\rho^e_k)}}\right)\quad \textrm{ with probability } p_s(\rho_k)=\tr{\Mc_s(\rho_k)}
.
$$
Then the fidelity $(F(\rho_k,\rho^e_k))$ is a $(\mathcal F_k)-$submartingale where $\mathcal F_k=\sigma\{(\rho_l,\rho^e_l)\vert l\leq k)\}$. In particular, we have
\begin{equation*}
\EE{F(\rho_{l},\rho^e_{l})|\mathcal F_k}=\EE{F(\rho_{l},\rho^e_{l})|(\rho_k,\rho^e_k)}\geq F(\rho_k,\rho^e_k),
\end{equation*}
 for all $l>k$.
\end{thm}
As far as we know, the continuous-time  asymptotic versions of the discrete-time dynamics~\eqref{eq:discrete-imp} and~\eqref{eq:ediscrete-imp}  associated to a left stochastic matrix $\eta^n$ with asymptotics~\eqref{eq:etan},  have not been established up to now.  In the following, we derive such    continuous-time SMEs which are the limits of these discrete-time dynamics and  prove their stability.
\subsection{Continuous-time filters as limit of discrete-time filters}~\label{sec:design}
Take  $n$ large  and   consider the piece-wise constant  continuous-time stochastic processes denoted by $\rho_n(t)$ and $\rho^e_n(t)$ with
\begin{equation}\label{eq:rhonrhoen}
\rho_n(t)=\rho_{[nt]},\quad\rho^e_n(t)=\rho^e_{[nt]}.
\end{equation}
where the discrete-time process $(\rho_{k},\rho^e_{k})$ obeys to~\eqref{eq:discrete-imp} with $k=[nt]$, the entire part of $nt$.
It is clear that $(\rho_n(t),\rho^e_n(t))$ is a Markov process, as $(\rho_k,\rho^e_k)$ is the Markov chain of Theorem~\ref{thm:discrete-imp}.

Now suppose that $\rho_t$ and $\rho^e_t$ be respectively the solutions of the continuous-time dynamics of the true filter and its estimate at time $t$. Let $\Ac_n$ and $\Ac$ be respectively  the Markov generators of $(\rho_n,\rho^e_n)$ and $(\rho,\rho^e)$. Then for all $C^2-$real valued function $f,$ we are looking for  the continuous-time processes $\rho$ and $\rho^e$ such that  the following limit holds.
$$
\lim_{n\rightarrow\infty}\sup_{(\rho,\rho^e)\in\mathcal D^2}\vert\mathcal A_n f(\rho,\rho^e)-\mathcal A f(\rho,\rho^e)\vert=0.
$$
We will see that such continuous-time  limit depends essentially on the structure of $\etaI=\lim_{n\rightarrow\infty} \eta^n$  described in Lemma~\ref{lem:etaI} and yields  a generalization of  the usual stochastic master equations driven by Wiener and Poisson processes.

\begin{lem}\label{lem:etaI}
\rm Take  the left stochastic matrix $\etaI_{s,(x,\mu)}$ defined by~\eqref{eq:etan} with line index $s\in\{1,\ldots,m\}$ and column index $(x,\mu)\in\{-1,1\}^{m_W}\times\{0,1,\ldots,m_P\}$. Consider the  following partition $(S^W,S^P)$ of  $\{1,\cdots,m\}$ labeling the number of experimental detections:
$$
S^P =\bigg\{ s\in\{1,\ldots,m\}\big| \sum_{x\in \{-1,1\}^{m_W}} \etaI_{s,(x,0)} =0 \bigg\}, \quad S^W =\bigg\{ s\in\{1,\ldots,m\}\big| \sum_{x\in \{-1,1\}^{m_W}} \etaI_{s,(x,0)} >0 \bigg\}.
$$
Then we have
\begin{itemize}

  \item The  $\etab_{s,\mu}$'s defined by
  \begin{equation}\label{eq:etabsmu}
    \etab_{s,\mu} \triangleq 2^{-m_W} \sum_{x\in \{-1,1\}^{m_W}}  \etaI_{s,(x,\mu)}, \quad (s,\mu)\in S^P\times\{1,\ldots,m_P\},
  \end{equation}
  satisfy
   \begin{equation}\label{eq:etabsmuC}
   \forall (s,\mu)\in S^P\times\{1,\ldots,m_P\}, ~ 0\leq \etab_{s,\mu}\leq 1 \text{ and }
    \forall \mu\in \{1,\ldots,m_P\}\sum_{s\in S^P} \etab_{s,\mu} \leq 1.
  \end{equation}

  \item The singular values of the  matrix $E$  with  entries
   \begin{equation}\label{eq:Eentry}
  \forall (s,\nu) \in S^W\times \{m_P,\ldots,m_P+m_W\},~    E_{s,\nu} \triangleq  2^{-{m_W}/2}
     \frac{\sum_{x}x_{\nu-m_P} \etaI_{s,(x,0)}}{\sqrt{\sum_{x} \etaI_{s,(x,0)}}}
  \end{equation}
belong to $[0,1]$.

\end{itemize}
\end{lem}

\begin{proof} Inequality~\eqref{eq:etabsmuC} is a direct consequence of $S^P$ definition, because for any $(x,\mu),$ $\sum_s \etaI_{s,(x,\mu)}=1$ and $x$  belongs to a set of cardinal $2^{m_W}$.

The matrix $E$ is well defined because, by definition of $S^W$, the denominators in~\eqref{eq:Eentry} are all strictly positive.
The singular values of $E$ do not exceed $1$, if and only if, for any unitary vector $z\in \mathbb{R}^{m_W}$, the Euclidian norm of $E z$ does not exceed $1$.
With $z=(z_j)$ we have
$$
\|E z\|^2 = 2^{-m_W}  \sum_{s\in S^W} \frac{\left(\sum_{x,j} x_j z_j \etaI_{s,(x,0)} \right)^2}{\sum_{x} \etaI_{s,(x,0)}}
$$
where  $(x,j)$ varies in $\{-1,1\}^{m_W} \times\{1,\ldots,m_W\}$.  With
$\vartheta_s =\sum_{x} \etaI_{s,(x,0)}$, $ \vartheta_{s,x}= \frac{\etaI_{s,(x,0)}}{\vartheta_s}$ and  $\bket{x|z} = \sum_j x_jz_j$
we have
$
\|E z\|^2 = 2^{-m_W}  \sum_{s\in S^W} \vartheta_s \left(\sum_{x} \bket{x|z} \vartheta_{s,x}  \right)^2
.
$
By convexity of $\alpha \mapsto \alpha^2$, we have
$$
\forall s \in S^W, \quad \left(\sum_{x} \vartheta_{s,x} \bket{x|z}   \right)^2\leq  \sum_{x} \vartheta_{s,x} \left( \bket{x|z}  \right)^2
$$
since $\sum_x \vartheta_{s,x}=1$ and $\vartheta_{s,x} \geq 0$.
Thus
 $$
 \|E z\|^2 \leq
  2^{-m_W}  \sum_{s\in S^W}  \etaI_{s,(x,0)} \sum_{x}  \left( \bket{x|z}   \right)^2
  =2^{-m_W} \sum_{x}  \left( \bket{x|z}   \right)^2
  $$
  since $\sum_s  \etaI_{s,(x,0)}= 1$ for any $x$. We have
$$
    \sum_{x}  \left( \bket{x|z}   \right)^2 = \sum_{x,j,j'} x_j x_{j'} z_j z_{j'}=  \sum_{j,j'} z_j z_j' \left(\sum_x x_j x_{j'}\right)
    \\
    = \sum_j z_j^2 \left(\sum_x x_j^2 \right)= 2^{m_W}\sum_j z_j^2
$$
since  $x$  in $\{-1,1\}^{m_W}$ implies that   $\sum_x x_j x_{j'}=0$ for $j\neq j'$. Thus  $\|E z\|^2 \leq 1$ when $\|z\| =1$.
\end{proof}

Next theorem provides a generalization of usual SME driven by Wiener processes with detection errors  to SME driven  simultaneously by Wiener and Poisson processes with detections errors.

\begin{thm}~\label{thm:pw}
\rm Consider $(\rho_n(t),\rho^e_n(t))$ defined by~\eqref{eq:rhonrhoen} and associated  to~(\ref{eq:discrete-imp},\ref{eq:ediscrete-imp})  with a left stochastic matrix $\eta^n$ verifying~\eqref{eq:etan}. Then, for $n\rightarrow +\infty$, the process  $(\rho_n(t),\rho^e_n(t))$  converges in distribution to the unique solutions of
\begin{multline}~\label{eq:Poisson and Wiener}
d\rho_t =-i[H,\rho_t] \,dt+ \left(\sum_\xi V_\xi\rho_t V_\xi^\dag - \half(V_\xi^\dag V_\xi\rho_t +\rho_t V_\xi^\dag V_\xi ) \right)\,dt
\\
+\sum_{s\in S^{P}}
   \left(\frac{\thetab_s\rho_t+\sum_{\mu}\etab_{s,\mu}V_\mu\rho_t V_\mu^\dag}{\thetab_s+\sum_{\mu}\etab_{s,\mu}\tr{V_\mu\rho_t V_\mu^\dag}} -\rho_t\right)
      \left(dN_s(t)-\Big( \thetab_s+\sum_{\mu}\etab_{s,\mu}\tr{V_\mu\rho_t V_\mu^\dag}\Big)\,dt\right)
\\
+\sum_{s\in S^{W}} \sqrt{\etab_s}
\left(\sum_{\nu} c_{s,\nu}\left( V_\nu\rho_t + \rho_t V_\nu^\dag - \tr{(V_\nu+V_\nu^\dag)\rho_t} \rho_t \right) \right)\,dW_s(t)
\end{multline}
and
\begin{multline}~\label{eq:cfepw}
d\rho^e_t =-i[H,\rho^e_t] +\left(\sum_\xi V_\xi\rho_t^e V_\xi^\dag - \half(V_\xi^\dag V_\xi\rho_t^e +\rho_t^e V_\xi^\dag V_\xi ) \right)\,dt
\\
+\sum_{s\in S^{P}}
    \left(\frac{\thetab_s\rho^e_t+\sum_{\mu}\etab_{s,\mu}V_\mu\rho^e_t V_\mu^\dag}{\thetab_s+\sum_{\mu}\etab_{s,\mu}\tr{V_\mu\rho^e_t V_\mu^\dag}} -\rho^e_t\right)
    \left(dN_s(t)- \Big( \thetab_s+\sum_{\mu}\etab_{s,\mu}\tr{V_\mu\rho^e_t V_\mu^\dag}\Big)\,dt\right)
\\
+\sum_{s\in S^{W}} \sqrt{\etab_s}
\left(\sum_\nu c_{s,\nu} \left( V_\nu\rho_t^e + \rho_t^e V_\nu^\dag - \tr{(V_\nu+V_\nu^\dag)\rho_t^e} \rho_t^e\right) \right) \times \ldots
\\
\ldots \times
\left(dy_s(t)-   \sqrt{\etab_s}\tr{\sum_\nu c_{s,\nu}(V_\nu+V_\nu^\dag)\rho^e_t}\,dt.\right)
\end{multline}
with  partition $(S_P,S_W)$ of $\{1,\ldots,m\}$ defined in Lemma~\ref{lem:etaI}, with,  in the above sums,  $\mu\in\{1,\ldots,m_P\}$, $\nu\in\{m_P+1,\ldots,m_P+m_W\}$ and $\xi\in\{1,\ldots, m_P+m_W\}$,  and where
\begin{itemize}

  \item  $s\in S^{P}$ is related to the Poisson process $dN_s(t)$ characterized by
$$
\langle dN_s(t)\rangle= \Big( \thetab_s+\sum_{\mu}\etab_{s,\mu}\tr{V_\mu\rho_t V_\mu^\dag}\Big)\,dt
$$
with  $\etab_{s,\mu}$ given by~\eqref{eq:etabsmu} and  $\thetab_s=2^{-m_W}\sum_x \etaT_{s,(x,0)}\geq 0$ ($\etaT$ defined in~\eqref{eq:etan}).

\item  $s\in S^{W}$ is related  to the continuous   signal $y_s$ related to the Wiener process $dW_s$ via
\begin{equation}~\label{eq:y-imp}
dy_s(t)=dW_s(t)+\sqrt{\etab_s} \tr{\sum_{\nu} c_{s,\nu}(V_\nu+V_\nu^\dag)\rho_t}\,dt.
\end{equation}
The efficiencies $\etab_s$ belong to  $[0,1]$ and correspond to  the eigenvalues of $E E^\dag$ with   matrix $E$ defined in Lemma~\ref{lem:etaI}. The
 real coefficients  $c_{s,\nu}$ are given by the entries of the orthogonal matrix $C$ appearing in the singular value decomposition of $E=R D C$ with $R$ and $C$ orthogonal and $D$ the rectangular diagonal matrix with diagonal entries  $\sqrt{\etab_s}$:
 $$
 \forall s,s' \in S^W,~\sum_\nu c_{s,\nu} c_{s',\nu}=\delta_{s,s'}
 $$
\end{itemize}
\end{thm}
The proof of this theorem is given in the next subsection. It admits the same structure as the proof of Proposition~\ref{pr:conv} but with sightly more complicated computations for the Markov generators. The stability of such quantum filters  is ensured in the following theorem.
\begin{thm}~\label{thm:pws}
  \rm Consider the Markov process $(\rho_t,\rho^e_t),$ satisfying~(\ref{eq:Poisson and Wiener},\ref{eq:cfepw}) where the positive integers $m_P$, $m_W$ and $m$ are arbitrary, where
 the $m_p+m_W$  square matrices $V_\xi$  are arbitrary, where the partition $(S^W,S^P)$ of $\{1,\ldots,m\}$ is   arbitrary, where $\forall (s,\mu)\in S^P\times\{1,\ldots,m_P\}$, $\thetab_s \geq 0$,   $\etab_{s,\mu}\in[0,1]$  and $\sum_{s'\in S^P} \etab_{s',\mu}\leq 1$, where
 $$
     \forall (s,s')\in S^W, \quad \etab_s \in[0,1] \quad \text{and}\quad
     \sum_{\nu\in\{m_P+1,\ldots,m_P+m_W\}} c_{s,\nu} c_{s',\nu}=\delta_{s,s'}.
 $$
 Then, the fidelity $F(\rho_t,\rho^e_t)$ defined in Equation~\eqref{eq:fidelity}  is a $(\mathcal F_t)-$submartingale, where $\mathcal F_t=\sigma\{(\rho_\tau,\rho^e_\tau)|\tau\leq t\}$. In particular, we have,  $$\EE{F(\rho_\tau,\rho^e_\tau)|\mathcal F_t}=\EE{F(\rho_\tau,\rho^e_\tau)|(\rho_t,\rho^e_t)}\geq F(\rho_t,\rho^e_t),$$ for all $\tau\geq t.$
\end{thm}
\begin{proof}
The proof is the similar  to the one given for Theorem~\ref{thm:continuous-time stability}. It relies on a direct application of  Theorem~\ref{thm:discrete-imp} and  Theorem~\ref{thm:pw}. The assumptions made on the real coefficients $\thetab_s$, $\etab_{s,_mu}$, $\etab_s$ and $c_{s,\nu}$ implies  the existence of    a family of stochastic process $(\rho_n(t),\rho_n^e(t))$ defined in~\eqref{eq:rhonrhoen} converging towards $(\rho_t,\rho^e(t))$ for $n$ large. This implication  relies on  manipulations based on Lemma~\ref{lem:etaI} and providing a  family of stochastic matrices  $\eta^n= \etaI + \etaT/n + O(1/n)$ such that  these coefficients  $(\thetab_s,\etab_{s,_mu},\etab_s,c_{s,\nu})$ are related to $\etaI$ and $\etaT$ according to  Theorem~\ref{thm:pw}.  These manipulations are simple and  not detailed here.
\end{proof}
\subsection{Proof of Theorem~\ref{thm:pw}}
The tightness property of $(\rho_n(t),\rho_n^e(t))$ can be concluded by the similar argument appeared in~\cite{pellegrini2010markov}. We give  here  the convergence proof of the Markov generators $\mathcal A_n$ of $(\rho_n(t),\rho_n^e(t))$   towards the Markov generator $\mathcal A$ of $(\rho_t,\rho^e_t)$.  We  detail the convergence  proof for
$S^P=\varnothing$ with  $S^W=\{1,\ldots,m\}$ and for $S^P=\{1,\ldots,m\}$ with $S^W=\varnothing$. The general case where both $S^P$ and $S^W$ are not empty is just a concatenation  of the two previous ones .
\subsubsection{Proof of Theorem~\ref{thm:pw} when  $S^P=\varnothing$ and   $S^W=\{1,\ldots,m\}$} \label{ssec:SW}
Consider the singular value decomposition of matrix $E$ defined in Lemma~\ref{lem:etaI}: $E=R D C$ with $R\in O(m)$, $C\in O(m_W)$ and $D$ the rectangular diagonal matrix with diagonal formed by $\sqrt{\etab_s}$. Notice that the entries of $C$ coincide with the  $c_{s,\nu}$.  Set $$
 \pb_s\triangleq 2^{-m_W} \sum_{x}\etaI_{s,(x,0)}, \quad
 \pbp_{s,\nu}\triangleq  2^{-m_W}\sum_{x | x_{\nu-m_P} =1}\etaI_{s,(x,0)}, \quad
 \pbm_{s,\nu}\triangleq  2^{-m_W}\sum_{x | x_{\nu-m_P} =-1} \etaI_{s,(x,0)}
 .
 $$
 Then $E_{s,\nu}= \sum_\nu \tfrac{\pbp_{s,\nu} -\pbm_{s,\nu}}{\bar p_s}$. Since $R$ is an orthogonal matrix, $dW=(dW_s)_{s\in\{1,\ldots,m\}}$ and  $R^\dag dW$ define similar Wiener processes. Thus, with replacing $dW$ by $R^\dag dW$ in~(\ref{eq:Poisson and Wiener},\ref{eq:cfepw}), we have
\begin{align}~\label{eq:cf}
d\rho_t&=-i[H,\rho_t] \,dt+ \sum_\xi \Lc_\xi(\rho_t)\,dt
+\sum_{s}\left(\sum_\nu \frac{\pbp_{s,\nu} -\pbm_{s,\nu}}{\sqrt{\pb_s}}\Lambda_\nu(\rho_t)\right)\,dW_s(t),
\\~\label{eq:cfe}
d\rho^e_t&=-i[H,\rho^e_t]\,dt + \sum_\xi \Lc_\xi(\rho^e_t)\,dt+\mathcal K(\rho_t,\rho^e_t)\,dt
+\sum_s  \sum_\nu \tfrac{\pbp_{s,\nu} -\pbm_{s,\nu}}{\sqrt{\pb_s}}\Lambda_\nu(\rho^e_t)dW_s(t),
\end{align}
with $$\mathcal K(\rho,\rho^e):=\sum_s\left(\sum_\nu \tfrac{\pbp_{s,\nu} -\pbm_{s,\nu}}{\bar p_s} \Big(\tr{(V_\nu+V_\nu^\dag)(\rho-\rho^e)}\Big)\right)\left(\sum_{\nu'}(\pbp_{s,\nu'} -\pbm_{s,\nu'})\Lambda_{\nu'}(\rho^e)\right).$$
The infinitesimal generator associated to the Markov process $(\rho_n,\rho^e_n)$ satisfying dynamics~\eqref{eq:discrete-imp} and~\eqref{eq:ediscrete-imp}, can be written as follows
$$
    \Ac_n f(\rho,\rho^e)\approx
    n \sum_s  p_{s}(\rho)\left( f\left( \tfrac{\Mc_s(\rho)}{p_{s}(\rho)},\tfrac{\Mc_s(\rho^e)}{p_{s}(\rho^e)}\right) -f(\rho,\rho^e)\right)
.
$$
We can approximate the generator $\Ac_n$ via  the second order expansion of $f$ around $\rho$ and $\rho^e$.
\begin{multline*}
    \Ac_n f(\rho,\rho^e)= n D_{(\rho,\rho^e)} f\cdot \left(\sum_s \Mc_s(\rho) - p_s(\rho) \rho,\sum_s \Big(\Mc_s(\rho^e) - p_s(\rho^e) \rho^e\Big)\frac{p_s(\rho)}{p_s(\rho^e)} \right)
    \\+ n \sum_{s} p_s(\rho)\tfrac{1}{2} D^2_{(\rho,\rho^e)} f\cdot \left(  \tfrac{\Mc_s(\rho)}{p_s(\rho)}-  \rho, \tfrac{\Mc_{s}(\rho^e)}{p_{s}(\rho^e)} -\rho^e;\tfrac{\Mc_s(\rho)}{p_s(\rho)}-  \rho, \tfrac{\Mc_{s}(\rho^e)}{p_{s}(\rho^e)}\right)
   + \mathcal O(n^{-1/2})
  .
\end{multline*}
Let us first calculate the first term in above. By using~\eqref{eq:MMxO}, ~\eqref{eq:PxO} and~\eqref{eq:MMPxmu}, we find
\begin{align}~\label{eq:msdyn}
\Mc_s(\rho)  -p_s(\rho) \rho&=
\sum_{x}\eta^n_{s,(x,0)} \left(M_{x,0}\rho M_{x,0}^\dag-p_{x,0}(\rho) \rho\right)
    + \sum_{x,\mu}\eta^n_{s,(x,\mu)} \left(M_{x,\mu}\rho M_{x,\mu}^\dag-p_{x,\mu}(\rho)  \rho\right)
  \nonumber\\ & =
    \sum_x \tfrac{ \eta^n_{s,(x,0)}c_W}{\sqrt{n}} \sum_\nu x_{\nu-m_P} \Lambda_\nu(\rho)
   \nonumber\\&+ \sum_x
    \tfrac{ \eta^n_{s,(x,0)}c_W}{n} \left(-i[H,\rho] + \sum_\xi \Lc_\xi(\rho)-\sum_\mu\Upsilon_\mu(\rho)\tr{V_\mu\rho V_\mu^\dag}\right)
   \nonumber \\
    &+  \sum_{x,\mu}\tfrac{\eta^n_{s,(x,\mu)}c_W}{n} \Upsilon_\mu(\rho)\tr{V_\mu\rho V_\mu^\dag}
    .
\end{align}
We have
\begin{equation}~\label{eq:detm}
    \sum_s \Mc_s(\rho) - p_s(\rho) \rho =
    \tfrac{1}{n} \left(-i[H,\rho] + \sum_\xi \Lc_\xi(\rho)\right)
    ,
\end{equation}
since for all $(x,0)$ and $(x,\mu)$: $\sum_s \eta^n_{s,(x,0)}=\sum_s \eta^n_{s,(x,\mu)}=1.$ Also note that for all $\nu:$ $$\sum_x x_{\nu-m_P} \Lambda_\nu(\rho)=0.$$
Equation~\eqref{eq:msdyn} can be rewritten as follows,
\begin{align}~\label{eq:mc}
\Mc_s(\rho) - p_s(\rho) \rho &=
    \sum_x \tfrac{ \eta^n_{s,(x,0)}c_W}{\sqrt{n}} \sum_\nu x_{\nu-m_P} \Lambda_\nu(\rho) + \mathcal O(n^{-1/2})
    \nonumber\\
    &= \sum_\nu \tfrac{\pbp_{s,\nu} -\pbm_{s,\nu} }{\sqrt{n}} \Lambda_\nu(\rho) + \mathcal O(n^{-1/2})
    .
\end{align}
Thus $\Mc_s(\rho) - p_s(\rho) \rho= \mathcal O(n^{-1/2}).$
Now using~\eqref{eq:PxO} and~\eqref{eq:MMPxmu}, we have
\begin{multline*}
    p_s(\rho)= \sum_{x}\eta^n_{s,(x,0)} p_{x,0}(\rho) + \sum_{x,\mu}\eta^n_{s,(x,\mu)} p_{x,\mu}(\rho)
    =
    \\
    \sum_{x}\eta^n_{s,(x,0)}  c_W \left(1 + \tfrac{1}{\sqrt{n}} \sum_\nu x_{\nu-m_P} \tr{  (V_\nu+V_\nu^\dag)\rho}-
    \tfrac{1}{n} \sum_\mu \tr{\rho V_\mu^\dag V_\mu} \right)
    \\
    +
    \sum_{x,\mu}\tfrac{\eta^n_{s,(x,\mu)}c_W}{n} \tr{V_\mu \rho V_\mu^\dag}
    .
\end{multline*}
By equation in above, the approximate value of $\tfrac{p_s(\rho)}{p_s(\rho^e)}$ is given by
\begin{align*}
\tfrac{p_s(\rho)}{p_s(\rho^e)}&= 1+\tfrac{c_W}{\sqrt{n}\bar p_s}\sum_x\eta^n_{s,(x,0)}\left(\sum_\nu x_{\nu-m_P} \Big(\tr{(V_\nu+V_\nu^\dag)\rho}-\tr{(V_\nu+V_\nu^\dag)\rho^e}\Big)\right)+\mathcal O(n^{-1/2})\\&= 1+\tfrac{1}{\sqrt{n}\bar p_s}\left(\sum_\nu\Big(\pbp_{s,\nu} -\pbm_{s,\nu}\Big) \Big(\tr{(V_\nu+V_\nu^\dag)\rho}-\tr{(V_\nu+V_\nu^\dag)\rho^e}\Big)\right)+\mathcal O(n^{-1/2}).
\end{align*}
Finally, we find the following
\begin{align*}
\sum_s n\Big(\Mc_s(\rho^e) - p_s(\rho^e) \rho^e\Big)\frac{p_s(\rho)}{p_s(\rho^e)}
=-i[H,\rho^e] + \sum_\xi \Lc_\xi(\rho^e)+\mathcal K(\rho,\rho^e)+\mathcal O(n^{-1/2}).
\end{align*}
Since $p_s(\rho)=\pb_s + \mathcal O(n^{-1/2})$ and
$\frac{\Mc_s(\rho)}{p_s(\rho)} -\rho = \mathcal O(n^{-1/2})$, the  zero order terms of $$n \sum_{s} p_s(\rho)\tfrac{1}{2} D^2_{(\rho,\rho^e)} f\cdot\left(  \tfrac{\Mc_s(\rho)}{p_s(\rho)}-  \rho, \tfrac{\Mc_{s} (\rho^e)}{p_{s}(\rho^e)} -\rho^e; \tfrac{\Mc_s(\rho)}{p_s(\rho)}-  \rho, \tfrac{\Mc_{s} (\rho^e)}{p_{s}(\rho^e)} -   \rho^e\right)$$ is given by the following computations:
{\small\begin{align*}
&n \sum_{s} p_s(\rho)D^2_{(\rho,\rho^e)} f\cdot
    \left(  \tfrac{\Mc_s(\rho)}{p_s(\rho)}-  \rho, \tfrac{\Mc_{s} (\rho^e)}{p_{s}(\rho^e)}-\rho^e;\tfrac{\Mc(\rho)}{p_s(\rho)}-\rho, \tfrac{\Mc_{s} (\rho^e)}{p_{s}(\rho^e)} -\rho^e\right)\\
     &= n \sum_{s} \tfrac{1}{\pb_s} D^2_{(\rho,\rho^e)} f\cdot\left(\Mc_s(\rho) - p_s(\rho) \rho,  \Mc_{s} (\rho^e) - p_{s}(\rho^e)  \rho^e;\Mc_s(\rho) - p_s(\rho) \rho,  \Mc_{s} (\rho^e) - p_{s}(\rho^e)  \rho^e\right) + \mathcal O(n^{-1/2})
    \\
    &= \sum_{s} \tfrac{1 }{\pb_s}
    D^2_{(\rho,\rho^e)} f\cdot\left( \sum_\nu (\pbp_{s,\nu} -\pbm_{s,\nu})\Lambda_\nu(\rho), \sum_{\nu} (\pbp_{s,\nu} -\pbm_{s,\nu})  \Lambda_{\nu}(\rho^e); \sum_{\nu'} (\pbp_{s,\nu'} -\pbm_{s,\nu'})\Lambda_{\nu'}(\rho), \sum_{\nu'} (\pbp_{s,\nu'} -\pbm_{s,\nu'})  \Lambda_{\nu'}(\rho^e)\right) \\
    &+ \mathcal O(n^{-1/2})
    ,
\end{align*}}
where we have used Equation~\eqref{eq:mc} and the fact that we have to keep only the zero order terms  of $\frac{1}{p_s(\rho)}$ and $\frac{p_s(\rho)}{p_s(\rho^e)^2}$ which is equal to $\frac{1}{\pb_s}$.

\medskip

We get finally the following expression for the generator $\Ac$ which is given as the limit of $\Ac_n$ when $n$ goes to infinity.
\begin{align*}
    &\Ac f(\rho,\rho^e)=
     D_\rho f\cdot \left(-i[H,\rho] + \sum_\xi \Lc_\xi(\rho),
     -i[H,\rho^e] + \sum_\xi \Lc_\xi(\rho^e)+\mathcal K(\rho,\rho^e)\right)
    \nonumber\\
    &+ \tfrac{1}{2}\sum_{s}
    D^2_{(\rho,\rho^e)} f\cdot\left( \sum_\nu \tfrac{\pbp_{s,\nu} -\pbm_{s,\nu}}{\sqrt{\pb_s}}\Lambda_\nu(\rho), \sum_{\nu} \tfrac{\pbp_{s,\nu} -\pbm_{s,\nu}}{\sqrt{\pb_s}}  \Lambda_{\nu}(\rho^e); \sum_{\nu'} \tfrac{\pbp_{s,\nu'} -\pbm_{s,\nu'}}{\sqrt{\pb_s}}\Lambda_{\nu'}(\rho), \sum_{\nu'} \tfrac{\pbp_{s,\nu'} -\pbm_{s,\nu'}}{\sqrt{\pb_s}}  \Lambda_{\nu'}(\rho^e)  \right).
\end{align*}
The expression given in above corresponds well to the infinitesimal generator associated to the Markov process $(\rho,\rho^e)$ satisfying dynamics~\eqref{eq:cf} and~\eqref{eq:cfe}.

Comparing this equation to the one given in~\eqref{eq:rho} and~\eqref{eq:rhoe} for perfect measurements, we observe that the Poisson processes have  completely disappeared. It just remains the ensemble average $\sum_\mu \Lc_\mu (\rho)$ in the deterministic part. Also, the original $m_W$ Wiener processes indexed by $\nu$ have been reorganized in $m$ Wiener processes indexed by $s$.
\subsubsection{Proof of Theorem~\ref{thm:pw} when  $S^P=\{1,\ldots,m\}$ and   $S^W=\varnothing$} \label{ssec:SP}
In this case, we have
\begin{multline}~\label{eq:Poisson}
d\rho_t=-i[H,\rho_t] \,dt+ \sum_\xi \Lc_\xi(\rho_t)\,dt
\\
+\sum_s\left(\frac{\thetab_s\rho_t+\sum_{\mu}\etab_{s,\mu}V_\mu\rho_t V_\mu^\dag}{\thetab_s+\sum_{\mu}\etab_{s,\mu}\tr{V_\mu\rho_t V_\mu^\dag}} -\rho_t\right)\left(dN_s(t)- \Big( \thetab_s+\sum_{\mu}\etab_{s,\mu}\tr{V_\mu\rho_t V_\mu^\dag}\Big)\,dt\right),
\end{multline}
and
\begin{multline}~\label{eq:Poissone}
d\rho^e_t=-i[H,\rho^e_t] \,dt+ \sum_\xi \Lc_\xi(\rho^e_t)\,dt
\\
+\sum_s\left(\frac{\thetab_s\rho^e_t+\sum_{\mu}\etab_{s,\mu}V_\mu\rho^e_t V_\mu^\dag}{\thetab_s+\sum_{\mu}\etab_{s,\mu}
  \tr{V_\mu\rho^e_t V_\mu^\dag}} -\rho^e_t\right)\left(dN_s(t)- \Big( \thetab_s+\sum_{\mu}\etab_{s,\mu}\tr{V_\mu\rho^e_t V_\mu^\dag}\Big)\,dt\right).
\end{multline}
The infinitesimal generator associated to the Markov process $(\rho_n,\rho^e_n)$  is given by the following
\begin{equation*}
\Ac_n f(\rho,\rho^e)\approx
    n \sum_s p_{s}(\rho)\left( f\left( \frac{\Mc_s(\rho)}{p_{s}(\rho)}, \frac{\Mc_s(\rho^e)}{p_{s}(\rho^e)} \right) -f(\rho,\rho^e)\right).
\end{equation*}
The equation in above can be rewritten as follows
\begin{multline}~\label{eq:genP}
\Ac_n f(\rho,\rho^e)\approx
    n \sum_s p_{s}(\rho)\left( f\left( \frac{\Mc_s(\rho)}{p_{s}(\rho)}, \frac{\Mc_s(\rho^e)}{p_{s}(\rho^e)} \right) -f(\rho,\rho^e)\right)\\
    +nD_{(\rho,\rho^e)}f\cdot\left(\sum_s \left(\Mc_s(\rho)-p_s(\rho)\rho\right),\sum_s \left(\Mc_s(\rho^e)-p_s(\rho^e)\rho^e\right)+\mathcal K_2(\rho,\rho^e)\right)\\
    -n\sum_s p_s(\rho)D_{(\rho,\rho^e)}f\cdot\left(\frac{\Mc_s(\rho)}{p_s(\rho)}-\rho,\frac{\Mc_s(\rho^e)}{p_s(\rho^e)}-\rho^e\right),
\end{multline}
with $\mathcal K_2(\rho,\rho^e):=\sum_s\left(\frac{\Mc_s(\rho^e)}{p_s(\rho^e)}-\rho^e\right)\left(p_s(\rho)-p_s(\rho^e)\right).$

\medskip

Now let us give the expression of $\Mc_s(\rho)$
\begin{align}~\label{eq:mcs}
\Mc_s(\rho)&=\sum_x \eta^n_{s,(x,0)} M_{x,0}\rho M_{x,0}^\dag+\sum_{x,\mu} \eta^n_{s,(x,\mu)}M_{x,\mu}\rho M_{x,\mu}^\dag\nonumber\\=
&\frac{1}{n}\left(\thetab_s\rho+\sum_{\mu}\etab_{s,\mu}V_\mu\rho V_\mu^\dag\right)+\mathcal O(n^{-\frac{3}{2}}),
\end{align}
where for the last term in above, we have used~\eqref{eq:MMxO},~\eqref{eq:MMPxmu}, ~\eqref{eq:etan}, $\thetab_s=c_W\sum_x \etaT_{s,(x,0)}$ and
$\etab_{s,\mu}= c_W \sum_x \etaI_{s,(x,\mu)}$.

Consequently, the probability $p_s(\rho)$ has the following form
\begin{equation}~\label{eq:ProP}
p_s(\rho)=\frac{1}{n}\left(\thetab_s+\sum_{\mu}\etab_{s,\mu}\tr{V_\mu\rho V_\mu^\dag}\right)+\mathcal O(n^{-\frac{3}{2}}).
\end{equation}
Now by using Equations~\eqref{eq:detm}, ~\eqref{eq:mcs} and~\eqref{eq:ProP}, we find the following expression for the limit of the generator $\Ac_n$ expressed in~\eqref{eq:genP}
{\small\begin{align*}
&\Ac f(\rho,\rho^e)=D_{(\rho,\rho^e)} f \cdot \left(-i[H,\rho] + \sum_\xi \Lc_\xi(\rho),-i[H,\rho^e] + \sum_\xi \Lc_\xi(\rho^e)+\mathcal K_2(\rho,\rho^e)\right)\\
&+\sum_s\left(   \thetab_s+\sum_{\mu}\etab_{s,\mu}\tr{V_\mu\rho V_\mu^\dag}\right)\times\\
&\left( f\left( \frac{\thetab_s\rho+\sum_{\mu}\etab_{s,\mu}V_\mu\rho V_\mu^\dag}{\thetab_s+\sum_{\mu}\etab_{s,\mu}\tr{V_\mu\rho V_\mu^\dag}}, \frac{\thetab_s\rho^e+\sum_{\mu}\etab_{s,\mu}V_\mu\rho^e V_\mu^\dag}{\thetab_s+\sum_{\mu}\etab_{s,\mu}\tr{V_\mu\rho^e V_\mu^\dag}}  \right) -f(\rho,\rho^e)\right)\\
&-\sum_s\left( \thetab_s+\sum_{\mu}\etab_{s,\mu}\tr{V_\mu\rho V_\mu^\dag}\right)\times\\
&D_{(\rho,\rho^e)} f\cdot \left(\frac{\thetab_s\rho+\sum_{\mu}\etab_{s,\mu}V_\mu\rho V_\mu^\dag}{\thetab_s+\sum_{\mu}\etab_{s,\mu}\tr{V_\mu\rho V_\mu^\dag}}-\rho,\frac{\thetab_s\rho^e+\sum_{\mu}\etab_{s,\mu}V_\mu\rho^e V_\mu^\dag}{\thetab_s+\sum_{\mu}\etab_{s,\mu}\tr{V_\mu\rho^e V_\mu^\dag}}-\rho^e\right).
\end{align*}}
In fact, we have $$\Ac_n f(\rho,\rho^e)=\Ac f(\rho,\rho^e)+\mathcal O(n^{-\half}).$$
We observe that the expression of $\Ac f(\rho,\rho^e)$ corresponds well to the generator associated to the Markov process $(\rho,\rho^e)$  satisfying SMEs~\eqref{eq:Poisson} and~\eqref{eq:Poissone}. Now use the tightness property of the sequence $(\rho_n,\rho^e_n)$ to conclude the convergence of the processes $(\rho_n,\rho^e_n)$ towards $(\rho,\rho^e)$.

\subsubsection{Proof of Theorem~\ref{thm:pw} when  $S^P\neq \varnothing$ and   $S^W\neq \varnothing$}

In this general case, the  Markov generator  of $(\rho_n,\rho^e_n)$  is decomposed into two different sums:
\begin{multline*}
     \Ac_n f(\rho,\rho^e)\approx
    n \sum_{s\in\{1,\ldots,m\}}  p_{s}(\rho)\left( f\left( \tfrac{\Mc_s(\rho)}{p_{s}(\rho)},\tfrac{\Mc_s(\rho^e)}{p_{s}(\rho^e)}\right) -f(\rho,\rho^e)\right)
\\
=    n \sum_{s\in S^W}   p_{s}(\rho)\left( f\left( \tfrac{\Mc_s(\rho)}{p_{s}(\rho)},\tfrac{\Mc_s(\rho^e)}{p_{s}(\rho^e)}\right) -f(\rho,\rho^e)\right)
+  n \sum_{s\in S^P}   p_{s}(\rho)\left( f\left( \tfrac{\Mc_s(\rho)}{p_{s}(\rho)},\tfrac{\Mc_s(\rho^e)}{p_{s}(\rho^e)}\right) -f(\rho,\rho^e)\right)
.
\end{multline*}
The limit for $n$ large of each sum versus $s\in S^W$ and $s\in S^P$ is then  obtained following exactly the same calculations as the ones given in subsections~\ref{ssec:SW} and~\ref{ssec:SP}.  Since~(\ref{eq:Poisson and Wiener},\ref{eq:cfepw}) are just the concatenation of
(\ref{eq:cf},\ref{eq:cfe}) for the Wiener part and (\ref{eq:Poisson},\ref{eq:Poissone})  for the Poisson part, we get directly   the convergence of the processes $(\rho_n(t),\rho^e_n(t))$ towards $(\rho_t,\rho^e_t)$.

%%%%%%%%%%%%%%%%%%%%%%%%%%%%%%%%%%%%%%%%%%%%%%%%%%%%%%%%%%%%%%%%%%%%%%%%%%%%%%%%%%%%%%%%
\section{Conclusion}~\label{sec:conclusion}
For a large class of SMEs driven by Wiener and Poisson processes,  Theorem~\ref{thm:pws} shows that  $\trr{\sqrt{\sqrt{\rho}\rho^e\sqrt{\rho}}}$   between the quantum state $\rho_t$ and its estimate $\rho^e_t$ is  a  submartingale.  Thus the "metric"  $1-\trr{\sqrt{\sqrt{\rho}\rho^e\sqrt{\rho}}}$ is a non-negative super-martingale that vanishes only when $\rho=\rho^e$. A natural question is the following: do there exist other "metrics" $D(\rho,\rho^e)$ that are super-martingales for such a large class of quantum systems?

 Any such  "metric" $D$ must be contractive  for all  Lindblad equations:  in Theorem~\ref{thm:pws},  there is no restriction on the degree of incompleteness of the measurements. Thus we can assume $\etab_s=0$, $S^{W}=\{1,\ldots,m\}$ and $S^{P}=\emptyset$. In this case,  $\rho$ and $\rho^e$ obey the same ordinary Lindblad differential equation
$$
\frac{d}{dt} \rho = -i[H,\rho] + \sum_\xi V_\xi\rho V_\xi^\dag - \half(V_\xi^\dag V_\xi\rho +\rho V_\xi^\dag V_\xi ),
$$
where $H$ and $V_\xi$ are arbitrary.
In~\cite{petz:LAA1996}, Petz has given, via the theory of operator monotone functions,  a complete characterization of distance  that are contractive for all Lindblad evolutions.  Could we exploit  Petz results  to characterize  "metrics" $D(\rho,\rho^e)$  that are  super-martingale for  all the  quantum filtering processes   considered in this paper?

%\bibliographystyle{plain}
%\bibliography{ref2}

\begin{thebibliography}{10}

\bibitem{AminiPhD}
H.~Amini.
\newblock {\em Stabilization of Discrete-Time Quantum Systems and Stability of
  Continuous-Time Quantum Filters}.
\newblock PhD thesis, Mines ParisTech, 2012.

\bibitem{aminicdc}
H.~Amini, M.~Mirrahimi, and P.~Rouchon.
\newblock On stability of continuous-time quantum-filters.
\newblock In {\em Proceedings of the 50th IEEE Conference on Decision and
  Control}, pages 6242--6247, 2011.

\bibitem{attal2006repeated}
St\'{e}phane Attal and Yan Pautrat.
\newblock From repeated to continuous quantum interactions.
\newblock In {\em Annales Henri Poincar\'{e}}, volume~7, pages 59--104.
  Springer, 2006.

\bibitem{Barchielli3}
A.~Barchielli and V.~P. Belavkin.
\newblock Measurements continuous in time and a posteriori states in quantum
  mechanics.
\newblock {\em Journal of Physics A: Mathematical and General}, 24(7):1495,
  1991.

\bibitem{Barchiellei1}
A.~Barchielli and M.~Gregoratti.
\newblock {\em Quantum Trajectories and Measurements in Continuous Time: the
  Diffusive Case}, volume 782.
\newblock Springer Verlag, 2009.

\bibitem{Bauer2012}
Benoist~T. Bauer, M. and D.~Bernard.
\newblock Repeated quantum non-demolition measurements: Convergence and
  continuous time limit.
\newblock {\em Annales Henri Poincar\'e}, 14(4):639679, 2013.

\bibitem{Belavkinone}
V.~P. Belavkin.
\newblock Quantum filtering of {M}arkov signals with white quantum noise.
\newblock {\em Ra-diotechnika i Electronika}, 25:1445–--1453, 1980.

\bibitem{belavkin1992quantum}
V.~P. Belavkin.
\newblock Quantum stochastic calculus and quantum nonlinear filtering.
\newblock {\em Journal of Multivariate Analysis}, 42(2):171--201, 1992.

\bibitem{Belavkin2007}
V.P. Belavkin.
\newblock Eventum mechanics of quantum trajectories: Continual measurements,
  quantum predictions and feedback control.
\newblock {\em arXiv:math-ph/0702079}, 2007.

\bibitem{Benoist2013}
T~Benoist and C.~Pellegrini.
\newblock Large time behavior and convergence rate for quantum filters under
  standard non demolition conditions.
\newblock {\em Communications in Mathematical Physics}, in press, 2013.

\bibitem{Vladimir}
V.~B. Braginsky and F.~Y Khalili.
\newblock {\em Quantum Measurement}.
\newblock Cambridge Univ Pr, 1995.

\bibitem{carmichael-book}
H.~. Carmichael.
\newblock {\em An Open Systems Approach to Quantum Optics}.
\newblock {Springer-Verlag}, 1993.

\bibitem{dalibard-et-al:PRL92}
J.~Dalibard, Y.~Castin, and K.~M{\o}lmer.
\newblock Wave-function approach to dissipative processes in quantum optics.
\newblock {\em Phys. Rev. Lett.}, 68(5):580--583, 1992.

\bibitem{davies1976quantum}
E.~B. Davies.
\newblock Quantum theory of open systems.
\newblock 1976.

\bibitem{gough2004stochastic}
J.~Gough and A.i Sobolev.
\newblock Stochastic {S}chr\"{o}dinger equations as limit of discrete
  filtering.
\newblock {\em Open Systems \& Information Dynamics}, 11(03):235--255, 2004.

\bibitem{haroche-raimond:book06}
S.~Haroche and J.-M. Raimond.
\newblock {\em Exploring the Quantum: Atoms, Cavities and Photons}.
\newblock Oxford University Press, New York, 2006.

\bibitem{pellegrini2007existence}
C.~Pellegrini.
\newblock Existence, uniqueness and approximation for stochastic
  {S}chr\"{o}dinger equation: the poisson case.
\newblock {\em arXiv preprint arXiv:0709.3713}, 2007.

\bibitem{pellegrini2008existence}
C.~Pellegrini.
\newblock Existence, uniqueness and approximation of a stochastic
  {S}chr\"{o}dinger equation: the diffusive case.
\newblock {\em The Annals of Probability}, pages 2332--2353, 2008.

\bibitem{pellegrini2010markov}
C.~Pellegrini.
\newblock {Markov chains approximation of jump-diffusion stochastic master
  equations}.
\newblock {\em Annales de l'Institut Henri Poincar\'{e}, Probabilit\'{e}s et
  Statistiques}, 46(4):924--948, 2010.

\bibitem{petz:LAA1996}
D.~Petz.
\newblock Monotone metrics on matrix spaces.
\newblock {\em Linear Algebra and its Applications}, 244:81--96, 1996.

\bibitem{rouchon2010fidelity}
P.~Rouchon.
\newblock Fidelity is a sub-martingale for discrete-time quantum filters.
\newblock {\em IEEE Transactions on Automatic Control}, 56(11):2743--2747,
  2011.

\bibitem{sayrin-et-al:nature2011}
C.~Sayrin, I.~Dotsenko, X.~Zhou, B.~Peaudecerf, T.~Rybarczyk, S.~Gleyzes,
  P.~Rouchon, M.~Mirrahimi, H.~Amini, M.~Brune, J.-M. Raimond, and S.~Haroche.
\newblock Real-time quantum feedback prepares and stabilizes photon number
  states.
\newblock {\em Nature}, 477(7362):73--77, 2011.

\bibitem{somaraju-et-al:acc2012}
A.~Somaraju, I.~Dotsenko, C.~Sayrin, and P.~Rouchon.
\newblock Design and stability of discrete-time quantum filters with
  measurement imperfections.
\newblock In {\em Proceedings of American Control Conference}, pages
  5084--5089, 2012.

\bibitem{Ramonthesis}
R.~van Handel.
\newblock {\em Filtering, Stability, and Robustness}.
\newblock PhD thesis, California Institute of Technology, 2006.

\bibitem{Ramon2009}
R.~van Handel.
\newblock The stability of quantum {M}arkov filters.
\newblock {\em Infinite Dimensional Analysis, Quantum Probability and Related
  Topics}, 12(1):153–--172, 2009.

\bibitem{wiseman-milburn:book}
H.~M. Wiseman and G.~J. Milburn.
\newblock {\em Quantum Measurement and Control}.
\newblock Cambridge University Press, 2009.

\end{thebibliography}

\end{document}